\documentclass[10pt,twocolumn]{article}
\usepackage[margin=0.75in]{geometry}
\usepackage{amsmath,amssymb,amsthm,mathtools}
\usepackage{array,graphicx,booktabs,xcolor,hyperref,microtype}
\usepackage{tikz}
\usetikzlibrary{arrows.meta,positioning,backgrounds}
\usepackage[font=small,labelfont=bf]{caption}
\hypersetup{colorlinks=true,linkcolor=blue!50!black,citecolor=blue!50!black,
            urlcolor=blue!50!black,
            pdftitle={TBR: Transport-Based Rendering with Deposition Strokes for Inverse Graphics},
            pdfauthor={Tianqi Liu, Yushan Han, Hang Liu}}
\graphicspath{{figs/}}

\newtheorem{proposition}{Proposition}
\newtheorem{lemma}{Lemma}
\newtheorem{definition}{Definition}
\theoremstyle{remark}\newtheorem*{remark}{Remark}
\newcommand{\R}{\mathbb{R}}
\newcommand{\Om}{\Omega}
\newcommand{\dist}{\operatorname{dist}}
\newcommand{\sgm}{\operatorname{sigmoid}}
\newcommand{\RA}{\mathcal{R}_A}
\newcommand{\RU}{\mathcal{R}_U}
\newcommand{\RD}{\mathcal{R}_D}

\title{\vspace{-2em}\textbf{TBR: Transport-Based Rendering with Deposition Strokes for Inverse Graphics}}
\author{Tianqi Liu$^{1}$ \quad Yushan Han$^{2}$ \quad Hang Liu$^{3,\ast}$\\[4pt]
\normalsize $^{1}$Independent Researcher, Beijing, China\\
\normalsize $^{2}$Huaxin Design Institute, Hangzhou, China\\
\normalsize $^{3}$Communication University of China, Beijing, China\\[2pt]
\normalsize $^{\ast}$Corresponding author: 3118@cuc.edu.cn}
\date{}

\begin{document}\maketitle

\begin{abstract}
We present a stroke design in which strokes are transport-coupled: each
stroke deposits material of its own area and moves every earlier mark
without changing its area, so later strokes deform earlier ones. We then
solve the inverse problem under this design: given a target image, we
optimise an ordered program of such strokes whose replay approximates it,
with digital marbling as the motivating medium. The stroke is a capsule that
continuously joins circular drops to drawn deposits; its transport is exactly
area-preserving, with a closed-form inverse outside the deposit, and a
variant with the same inverse differs from line-source potential flow by
$8\%$ of the mean displacement. A replay
adjoint regenerates intermediate states instead of storing them and uses
$8.7\times$ less memory than checkpointed automatic differentiation; a fused
implementation fits a 2000-stroke program at $1024^2$ in about four minutes
on one GPU. On five marbled sheets the recovered programs are level with a published
stroke-based fitter as rasters, replay across a fourfold resolution range, and support edits in program order and palette space that
stay valid under transport.
\end{abstract}

\section{Introduction}
Turning an image into a program is useful when the program, rather than a
raster approximation, is the desired output: it can be replayed at a new
resolution, edited at the level of actions, or supplied to an authoring or
fabrication system. Stroke-based rendering provides this representation for
media in which a placed mark stays where it is put: later marks may cover it,
but do not change its geometry \cite{Hertzmann2003,Zou2021,Liu2021}. That
assumption is productive for drawing and painting, and it removes the
defining operation of marbling, the displacement of existing pigment.

\begin{figure*}[t]\centering
\includegraphics[width=\textwidth]{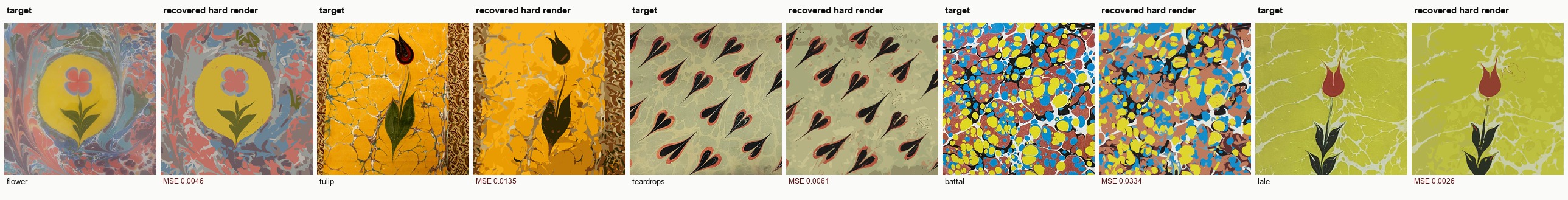}
\caption{\textbf{Target-to-program recovery.} Each pair shows a reference
image and the hard render obtained by replaying its recovered 2000-stroke
deposit-and-transport program. The five references cover both repeated motifs
and all-over patterns; no panel is selected by score.}\label{fig:recovery}
\end{figure*}

In paper marbling, pigment floats on a liquid surface; a new deposit creates
paint and displaces the pattern already present, so the location of an early
mark in the final image depends on every later stroke. Order matters, local
corrections propagate globally, and a program must be evaluated as a coupled
process. The inverse problem is to recover one program whose execution is
visually equivalent to a given sheet; the historical sequence that produced
the sheet is not the target.

Computational marbling provides forward simulators and composable closed-form
maps \cite{Acar2006,Lu2012,Jaffer2018,James2026}, authored by a user rather
than recovered from a target; its coupled primitives are either a disc (the
drop) or a drag that moves pigment without depositing it (the tine and
stylus). Differentiable simulation recovers continuous fields rather than a
discrete, ordered program of typed strokes \cite{McNamara2004,Hu2020}.
Stroke-based rendering optimises marks that stay put, and its differentiable
smudge operator \cite{Jiang2025} blends existing canvas colour along the brush
rather than transporting material. We develop deposition strokes with exactly
area-preserving exterior transport and a closed-form inverse, which admit
replay differentiation through long stroke sequences, together with the
inverse-graphics solution under that design: each stroke adds material of its
own area and moves the existing pattern without changing its area. The stylus strokes of the craft
are outside the family.

Fitting requires differentiation through the whole stroke chain, and storing
every intermediate state at this scale takes tens of gigabytes. Because the
transport map is invertible outside each deposit, a replay adjoint
regenerates the exterior states \cite{Gomez2017,Maclaurin2015,Vicini2021} and
records only the coordinates lost inside deposits, which a geometric
obstruction shows to be the single exact loss. The stroke itself is a capsule
whose length joins a circular drop to a drawn deposit continuously; its
transport is a translation in parallel-set area coordinates, which gives the
forward map and its inverse in closed form. Geometry, width and palette
assignment are optimised jointly through 2000 strokes.

We evaluate the recovered programs against controls that isolate transport: a
flat-colour floor, a matched ablation with transport disabled, a one-pass
closed-form correction, a second transport operator built from a published
primitive, three seeds, and synthetic targets with known generating
programs. A stroke-based renderer fits the same images but optimises marks
that stay put, so we run one at matched stroke count and report where it
wins. The output is an executable program that replays unchanged across a
$4\times$ resolution range and can be edited at the level of strokes and
stages (Fig.~\ref{fig:interv}).

\paragraph{Contributions.}
\begin{enumerate}\itemsep2pt
\item A transport-coupled deposition stroke for digital marbling: a capsule
that is exactly area-preserving off its footprint, has a closed-form inverse,
joins drops to drawn deposits continuously, and has a measured discrepancy
from line-source potential flow together with a variant that reduces it
(Sec.~\ref{sec:model}).
\item A replay adjoint for this recurrence, with the single exact obstruction
characterised and its record measured, using $8.7\times$ less memory than
checkpointed automatic differentiation (Sec.~\ref{sec:adjoint}).
\item A system that recovers executable 2000-stroke programs from single
images in minutes, with program-level edits, documented failure cases, and an
evaluation that separates image reconstruction from generating-stroke
correspondence (Secs.~\ref{sec:exp} and~\ref{sec:process}).
\end{enumerate}

The program is executable with respect to the computational stroke model:
``execution'' below means replay by its hard renderer. The operator is a geometric transport model rather than a simulation of a
bath; Sec.~\ref{sec:fidelity}
quantifies its discrepancy from a line-source potential-flow reference, itself
a transport model.

\section{Related work}
\paragraph{Forward models of marbling.}
Computational marbling has primarily been an authoring and simulation problem.
Fluid solvers \cite{Acar2006} reproduce the response of floating pigment to user input. Mathematical marbling replaces a full fluid solve with composable
maps \cite{Lu2012}; the classical drop map is closed-form and area-preserving.
Jaffer \cite{Jaffer2018} derives Oseen-flow maps for short stylus motion. Most
recently, James and James \cite{James2026} derive sharp analytical fluid
brushes from potential flow around a cylindrical tine, with reverse-drift
functions that encode tine insertion, motion and removal, compose in a
shader, and render progressively at any resolution. The tine, comb and stylus tools of these systems move pigment without depositing it, and deposition enters as the circular drop.
Inserting a tine of radius $a$ into an incompressible sheet displaces material
by $\sigma=\sqrt{\sigma'^2-a^2}$; that is the classical drop map, and it is
identically the $\ell=0$ limit of \eqref{eq:back}, so at the drop the two
constructions are the same map. The two constructions extend the drop map in different directions.
Mixwell adds tine motion and removal, which transport material without net
injection: the composed brush has no deposition set, is invertible
everywhere, and raises no obstruction of the kind
Lemma~\ref{lem:obstruction} identifies. We add extent instead --- the capsule
carries the drop continuously to an elongated deposit, keeping the closed-form
inverse --- and it is that continued deposition which makes the chain
non-injective on its footprint and gives the adjoint of
Sec.~\ref{sec:adjoint} something it must record; tine-style drag is the zero-injection
class we set aside (Sec.~\ref{sec:limits}). All of these works establish
efficient forward construction; none fits a program to an image. Our task
starts from the finished image and searches for the ordered construction
program.

\paragraph{Inverse design through physical dynamics.}
Adjoint fluid control \cite{McNamara2004} and differentiable programming systems \cite{Hu2020} show that desired outcomes can be converted into controls
through multi-step simulation, and recent work derives adjoints directly on
flow maps to differentiate long simulations accurately \cite{LiAdjoint2025}.
That work also regenerates state rather than storing it, but for a flow that
is reversible; the distinction here is that our map is deliberately not
reversible on the deposition set, and the contribution is to characterise and
measure that exception rather than to avoid it. These methods optimise continuous fields,
initial conditions, or controls. Our replay strategy instead specialises
state reconstruction to a stroke map that loses coordinates on its deposition set. Here the unknown
is a discrete-length program of typed pigment-creating strokes, and
the rendered appearance is carried by material deposited at different times.

\paragraph{Images as action programs.}
Stroke-based rendering spans hand-designed placement
\cite{Hertzmann2003}, sequential agents
\cite{Huang2019}, differentiable rasterisation \cite{Li2020,Zou2021}, and feed-forward stroke
prediction \cite{Liu2021}; the line begins with Haeberli's ordered marks
\cite{Haeberli1990} and Hertzmann's multi-size layering \cite{Hertzmann1998},
which uses a coarse-to-fine radius schedule similar to ours. A parallel line recovers
executable \emph{programs} from images --- hand-drawn diagrams
\cite{Ellis2018}, procedural models fitted by search \cite{Talton2011} --- where
the program's semantics is a static composition; here it is a coupled
dynamical process, so the contribution of slot $i$ depends on every later
slot. These methods composite marks \cite{Porter1984}; their
strength is direct image approximation rather than subsequent transport of
the accumulated canvas. Physically based paint systems model wet media during
forward interaction \cite{Curtis1997},
while robotic painters incorporate medium response through calibration or
feedback \cite{Schaldenbrand2023}.

A separate line recovers or synthesises a plausible \emph{process}: a drawing
order for line drawings \cite{Fu2011}, ordered translucent strokes from
time-lapse paintings \cite{Tan2015}, and ordered layered paths from one image
\cite{Ma2022}. These methods use different structural and learned priors;
they do not fit the deposition-and-transport model studied here. We optimise
stroke parameters in fixed ordered slots, so later transport affects the
image contribution of earlier slots; we do not search over discrete order.

The closest recent work is Jiang et al.~\cite{Jiang2025}, which fits
brushstrokes differentiably and includes smudging, so later strokes can alter
paint already placed. Its smudge renderer samples and blends existing canvas
colours along brush motion. Our construction instead uses explicit material
coordinate maps that preserve area outside deposits and admit exterior
inversion. The contribution is the capsule-deposition primitive and an
adjoint specialised to reconstructing its transport trajectories while
retaining coordinates lost by the deposition extension. Both approaches
produce structured painting representations; our replay and editing results
are evaluated within the specified marbling model.

\paragraph{Memory-efficient differentiation.}
Long reverse-mode chains can store all states, checkpoint them
\cite{Griewank2000}, or regenerate them by reversing the forward
step \cite{Maclaurin2015,Gomez2017}. Related replay strategies
are central to differentiable light transport
\cite{Vicini2021}, and flow-map adjoints obtain the
state-reconstruction benefits for long fluid simulations by reusing
flow maps \cite{LiAdjoint2025}; solver-level adjoints run an executed
block-implicit sweep backwards with a per-sweep replay buffer in place of
an autograd tape \cite{Shu2026}. Where reversal loses information,
reversible learning stores exactly the discarded bits \cite{Maclaurin2015};
our recurrence specialises that principle to pigment transport. The material a
stroke creates has no earlier state to regenerate; we construct that
exceptional set explicitly, record it, and measure its size. The soft coverage used during optimisation follows differentiable
rasterisation \cite{LiuSoftRas2019}; edge sampling \cite{LiEdge2018} and
reparameterisation \cite{Loubet2019} are the unbiased alternatives for
gradients across a visibility discontinuity; we use the annealed relaxation,
and every reported number is a hard render.

\section{Problem formulation}\label{sec:problem}
Let $\Om=[0,1]^2$ be the canvas and $I^\ast:\Om\to[0,1]^3$ the target image. A
\emph{program} is an ordered sequence $\Pi=(g_1,\dots,g_N)$ of strokes
$g_i=(a_i,d_i,r_i,c_i)$ with anchor $a_i\in\R^2$, offset $d_i\in\R^2$,
half-width $r_i>0$ and colour $c_i\in[0,1]^3$. Stroke $i$ has \emph{footprint}
\begin{equation}
K_i=\{x\in\R^2:\ \dist(x,[a_i,a_i+d_i])\le r_i\},
\end{equation}
the closed capsule of radius $r_i$ about its segment; $\ell_i=\lVert d_i\rVert=0$
is a disc, $\ell_i>0$ a drawn line, and stroke type is the continuous
parameter $\ell_i$.

\begin{table}[t]\centering\footnotesize\setlength{\tabcolsep}{4pt}
\begin{tabular}{@{}l@{\hskip 5pt}p{0.63\columnwidth}@{}}\toprule
symbol & meaning (first use)\\
\midrule
$\Pi=(g_1,\dots,g_N)$ & program: $N$ strokes in execution order\\
$g_i=(a_i,d_i,r_i,c_i)$ & anchor, offset, half-width, colour\\
$\ell_i=\lVert d_i\rVert$ & stroke length; $0$ is a drop, $>0$ a drawn line\\
$K_i$ & footprint: capsule of radius $r_i$ about the segment\\
$\Om=[0,1]^2$, $P$ & canvas and its pixel count\\
$\psi_i$ & backward material map of stroke $i$ (Eq.~\ref{eq:back})\\
$\Phi_i$ & exterior inverse of $\psi_i$ (Eq.~\ref{eq:fwd})\\
$\mathcal{R}_A,\mathcal{R}_U$ & renderer with transport applied / omitted\\
$\mathcal{R}_D$ & third renderer, $16$ area-matched drops (Sec.~\ref{sec:transfer})\\
$\sigma$ & distance from a point to the stroke's segment\\
$A_\ell,P_\ell$ & area, perimeter of the parallel set at $\sigma$ (Eq.~\ref{eq:steiner})\\
$u=A_\ell(\sigma)$ & area coordinate; $\Phi$ is a translation in $u$\\
$q$ & arclength along the iso-distance contour, normalised\\
$v$ & $q$ re-centred on the flank midpoint (Prop.~\ref{prop:tan})\\
$\alpha_i,\;\tau$ & coverage of stroke $i$; relaxation width (Eq.~\ref{eq:relax})\\
$\mathsf{P}$ & trainable palette, $8$ colours (Sec.~\ref{sec:param})\\
$s_i=(x_i,C_i,T_i)$ & per-pixel state: position, colour, transmittance\\
$\lambda^C_i,\lambda^T_i,\lambda^x_i$ & adjoints of $\mathcal{L}$ with respect to $s_i$ (Eq.~\ref{eq:adjstep})\\
$\lambda$ & weight on the length regulariser (Eq.~\ref{eq:inverse})\\
\bottomrule
\end{tabular}
\caption{Notation. Subscript $i$ indexes strokes throughout; the render
processes them from $N$ down to $1$, so state indices \emph{decrease} as
rendering proceeds.}\label{tab:notation}
\end{table}

A forward model $\mathcal{R}$ maps a program to an image. We consider models
of the following form. Each stroke $i$ induces a map $\psi_i:\R^2\to\R^2$
that traces pre-existing material from after stroke $i$ to before it.
Inside newly deposited material, it is a computational extension rather than
a material pre-image. For a pixel $x\in\Om$ set $x_N=x$ and, for
$i=N,\dots,1$,
\begin{equation}
\alpha_i=\alpha(x_i;g_i),\qquad
x_{i-1}=\psi_i(x_i),
\label{eq:chain}
\end{equation}
where $\alpha(\cdot\,;g_i)$ is the coverage of stroke $i$ ($\mathbf{1}_{K_i}$
for a hard render). The rendered colour is back-to-front compositing
\cite{Porter1984} along this chain,
\begin{equation}
\mathcal{R}(\Pi)(x)=\sum_{i=1}^{N}c_i\,\alpha_i\!\prod_{j>i}(1-\alpha_j)
\;+\;c_0\prod_{j=1}^{N}(1-\alpha_j),
\label{eq:render}
\end{equation}
with $c_0$ the background. Equation~\eqref{eq:render} is ordinary stroke
compositing; within this fixed stroke and coverage family, the renderers differ through $\psi_i$.
Evaluating coverage at a back-traced position is the semi-Lagrangian view of
advection: the coverage of stroke $i$ is evaluated where the pixel's material
\emph{was} when stroke $i$ acted, and that position depends on every later
stroke.

\begin{definition}\label{def:aware}
A forward model of the form \eqref{eq:chain}--\eqref{eq:render} is
\emph{inert-mark} if $\psi_i=\mathrm{id}$ for every $i$, and \emph{displacing}
otherwise. A displacing model is specified by the family of material maps
$\{\psi_i\}$; the definition is purely structural.
\end{definition}

The definition separates the structural property needed by the inverse
problem from the choice of a particular transport model. We restrict to
responses determined by the stroke alone: $\psi_i$ does not depend on the
current pigment loading. The exterior transport maps for insertions at the same segment commute
because the operator below translates an area coordinate. Coloured deposition
operations need not commute, since their coverage is ordered. Insertions at different
segments generally do not, because the second acts on material already moved
by the first.

The inverse problem is
\begin{equation}
\min_{\theta}\ \mathcal{L}\big(\mathcal{R}(\Pi(\theta)),I^\ast\big)
+\lambda\,\mathcal{P}(\Pi(\theta)),
\label{eq:inverse}
\end{equation}
over a smooth parameterisation $\theta\mapsto\Pi(\theta)$ that enforces the
feasible set by construction (Sec.~\ref{sec:param}), with data term
$\mathcal{L}$ and regulariser $\mathcal{P}$. Our task is to obtain a program attaining low reconstruction loss under the displacing renderer. The inert-mark case and a closed-form
pre-image correction serve as comparison methods; neither defines the
target task.

\paragraph{Required properties.} The construction relies on five properties. Four are
properties of a single map: (P1) a closed-form backward map; (P2) a computable
inverse (closed-form here); (P3) preservation of area outside newly deposited
material; and (P4) non-injectivity confined to the deposition set $K_i$. P1
evaluates the renderer, P2 regenerates states, P3 provides incompressible area
bookkeeping, and P4 identifies the trajectory samples to record; their number
is measured rather than bounded by deposited area, and computable iterative
inverses can also support replay (Sec.~\ref{sec:adjoint}). The fifth is a
property of the \emph{family}: (P5) stroke type varies continuously within
it, and every member satisfies P1--P4. Without P5 the choice between a drop
and a drawn deposit is combinatorial and has to be searched; with it, type is
one more differentiable coordinate. The operator below has all five, and
The supplement (Sec.~S10) explains what each one excludes.

\section{The capsule insertion operator}\label{sec:model}
A stroke injects new material and pushes the existing surface outward
without compressing or stretching it. We construct $\psi_i$ so that this
holds exactly, in three steps: coordinates wrapped around the stroke, an
area law in those coordinates, and a map that adds area. Fix a stroke with
segment of length $\ell$ and half-width $r$. Every point off the segment has
a distance $\sigma(x)=\dist(x,\text{segment})$ and lies on one parallel curve
$\{\sigma=\text{const}\}$, a stadium around the segment; let $q\in[0,1)$ be
its arclength along that curve as a fraction of the curve's perimeter. By
Steiner's formula the region enclosed by the parallel curve at distance
$\sigma$ has area and perimeter
\begin{equation}
A_\ell(\sigma)=2\ell\sigma+\pi\sigma^2,\qquad
P_\ell(\sigma)=2\ell+2\pi\sigma=A_\ell'(\sigma).
\label{eq:steiner}
\end{equation}
Moving outward by $d\sigma$ sweeps a band of area $P_\ell(\sigma)\,d\sigma$,
so $dA=P_\ell(\sigma)\,d\sigma\,dq$. The \emph{area coordinate}
$u=A_\ell(\sigma)$ absorbs that factor: $dA=du\,dq$, and a patch of the
plane has the same area whether it is measured in $(x,y)$ or in $(u,q)$.
The chart is defined off the segment; the supplement records its regularity.

\paragraph{Insertion.} The stroke inserts material of area $A_\ell(r)$. In
the area chart, pushing the surface outward is addition:
\begin{equation}
\Phi:(u,v)\longmapsto\big(u+A_\ell(r),\,v\big),\qquad u>0 ,
\label{eq:phi}
\end{equation}
where $v=q-\ell/(2P_\ell(\sigma))$ is arclength counted from the midpoint of
the capsule's flank rather than from a corner of it. In words, $u$ is the
area enclosed between the segment and the parallel curve through the point,
$v$ is the point's position along that curve as a fraction of its length,
and the stroke moves every point to the parallel curve that encloses
$A_\ell(r)$ more area, at the same position along it.

\begin{proposition}[Measure preservation]\label{prop:area}
$\Phi$ is measure-preserving and maps the plane minus the segment bijectively
onto the exterior of the footprint, $\{\sigma>r\}$.
\end{proposition}
Since $dA=du\,dv$ and $\Phi$ is a translation in $u$, its Jacobian
determinant is one and the statement follows; the supplement records the
argument. The exterior backward map is $\psi=\Phi^{-1}$: subtract the area,
$A_\ell(\sigma)=A_\ell(\sigma')-A_\ell(r)$, and solve the quadratic
\eqref{eq:steiner} for the distance. Each direction is written as the root
that is free of cancellation:
\begin{align}
\sigma&=\frac{C}{\ell+\sqrt{\ell^2+\pi C}},\qquad
C=(\sigma'-r)\big(2\ell+\pi(\sigma'+r)\big),\label{eq:back}\\
\sigma'&=\frac{B}{\ell+\sqrt{\ell^2+\pi B}},\qquad
B=2\ell(\sigma+r)+\pi(\sigma^2+r^2).\label{eq:fwd}
\end{align}
Here $C=A_\ell(\sigma')-A_\ell(r)$ and $B=A_\ell(\sigma)+A_\ell(r)$.
Equation~\eqref{eq:back} renders; \eqref{eq:fwd} is its exact inverse and is
what the adjoint of Sec.~\ref{sec:adjoint} iterates. The factored $C$ matters
where $\sigma'\to r$, at the pixels adjacent to a fresh stroke; the conjugate
form of \eqref{eq:fwd} matters for long thin strokes, $\ell/\pi\gg\sigma,r$.
At $\ell=0$, \eqref{eq:back} is the drop map $\sigma=\sqrt{\sigma'^2-r^2}$
of \cite{Lu2012}.

\paragraph{The tangential coordinate.} Equation~\eqref{eq:phi} fixes how far
a point moves outward. Which position along the parallel curve it keeps is
a separate choice, since a shift along the curve that depends only on $u$
preserves area as well; symmetry decides it.

\begin{proposition}[Centred arclength]\label{prop:tan}
Every map $(u,q)\mapsto(u+A_\ell(r),\,q+\delta(u))$ is measure-preserving;
among families that depend continuously on the inserted area and reduce to
the identity at zero area, the one that
preserves the centred coordinate $v$ is the only one that commutes with the
reflection of the capsule across its perpendicular bisector.
\end{proposition}
The supplement gives the symmetry argument; Table~\ref{tab:detj} measures
the error made by keeping the uncentred coordinate $q$ instead. The tangential motion is thus
fixed by symmetry rather than by a flow model: in potential flow, injection
moves material along the gradient of a harmonic potential, whose level sets
are not the parallel curves of the segment except at $\ell=0$.
Sec.~\ref{sec:fidelity} measures the difference.

\paragraph{Choice of skeleton.} The skeleton is a straight segment rather
than a curve for one reason. The construction needs the parallel curves of
the skeleton and the area they enclose in closed form, and the area law has
to hold at every distance, because back-traced coordinates travel
arbitrarily far from the canvas. Around a curved skeleton the parallel
curves run into one another on the concave side once $\sigma$ exceeds the
skeleton's \emph{reach}, and the Steiner count then over-counts the overlap:
for a circular arc of radius $R$ by several per cent at $\sigma=2R$. The law
holds at every distance only if the reach is infinite, which by Federer's
characterisation \cite{Federer1959} means a convex skeleton, and the only
compact convex curves are segments and points. A circular arc admits
closed-form projection and elementary arclength but is exact only below its
reach $R$; a cubic B\'ezier fails on the other two counts as well, its
nearest-point projection being the root of a quintic and its arclength an
elliptic integral. Curved skeletons, with the
approximation a finite reach implies, are left to future work;
Sec.~\ref{sec:transfers} uses the same condition for another medium.

\begin{table}[t]\centering\scriptsize\setlength{\tabcolsep}{3pt}
\begin{tabular}{l rr}\toprule
property (max over 200{,}000 exterior pts) & float32 & float64\\
\midrule
round trip $\Phi\circ\psi=\mathrm{id}$ & $1.9\times10^{-6}$ & $5.7\times10^{-15}$\\
area Jacobian $|\det J-1|$ & $1.1\times10^{-6}$ & $1.6\times10^{-15}$\\
drop limit, $\ell=10^{-6}$ & $2.7\times10^{-6}$ & $2.9\times10^{-6}$\\
midplane reflection, centred $v$ & $2.7\times10^{-6}$ & $5.9\times10^{-15}$\\
midplane reflection, raw $q$ & $4.6\times10^{-2}$ & $4.6\times10^{-2}$\\
\bottomrule
\end{tabular}
\caption{Numerical verification of Propositions~\ref{prop:area}
and~\ref{prop:tan} and of Equations~\eqref{eq:back}--\eqref{eq:fwd}, in canvas
units, on a capsule with $\ell=0.10$, $r=0.03$. The Jacobian is evaluated by
exact differentiation of the implemented map, so $\det J=1$ holds to machine
precision, as Proposition~\ref{prop:area} requires. The last row is the same
test with the arclength origin at a flank corner rather than the flank
midpoint: preserving raw $q$ breaks the capsule's mirror symmetry by four
orders of magnitude more than the discretisation, which is what
Proposition~\ref{prop:tan} rules out.}\label{tab:detj}
\end{table}

\paragraph{The chain.} With $\psi_i$ given by \eqref{eq:back} in stroke
$i$'s coordinates, extended to $K_i$ by clamping $C$ to zero --- which sets
$\sigma=0$ and preserves the centred arclength, so a point at signed
arclength $t$ from the flank midpoint, at distance $\sigma'$ on the flank,
goes to position $\ell t/(\ell+\pi\sigma')$ on the segment; this extension
is continuous across $\partial K_i$, where it agrees with the exterior
limit, differs from nearest-point projection, and is the one the
implementation computes --- \eqref{eq:chain}--\eqref{eq:render} define the
transport renderer $\RA$; setting $\psi_i=\mathrm{id}$ gives the inert-mark
model $\RU$. Both are evaluated by the same recurrence. The operator acts on
the free plane: there is no wall, so a bounded tray with a no-flux boundary
is not modelled (Sec.~\ref{sec:limits}).

\subsection{Fidelity of the surrogate}\label{sec:fidelity}
To quantify the discrepancy between the operator and the line-source
potential-flow reference, and how far
the evaluator of Sec.~\ref{sec:transfer} is from both, we use one functional.
For two backward maps of the same stroke, on the parallel curve
$\Gamma=\{\sigma'=2r\}$ one half-width outside the footprint after
insertion,
\begin{equation}
D(\psi_1,\psi_2)=\frac{\max_{x\in\Gamma}\lVert\psi_1(x)-\psi_2(x)\rVert}
{\operatorname{mean}_{x\in\Gamma}\lVert x-\psi_1(x)\rVert},
\label{eq:D}
\end{equation}
the largest disagreement between the two pre-images relative to the mean
displacement $\psi_1$ applies there. As reference we take $\psi_L$:
quasi-static injection of the same area $A_\ell(r)$ by a uniform line source
on the segment in 2-D potential flow, whose velocity is elementary
(logarithms and arctangents) and whose particle paths we integrate backward
for the injection time. All three maps insert the same area, so $D$ measures
foliation and tangential transport, not mass. Table~\ref{tab:foliation}: at
the median aspect ratio of the recovered strokes ($\ell/r\approx2$) the capsule operator differs from $\psi_L$ by $27\%$ under $D$, saturating near $30\%$ for long strokes and vanishing as $\ell\to0$, where it is exact.

\begin{table}[t]\centering\footnotesize\setlength{\tabcolsep}{5pt}
\begin{tabular}{r rrr}\toprule
$\ell/r$ & $D(\psi_A,\psi_L)$ & $D(\psi_D,\psi_L)$ & $D(\psi_A,\psi_D)$\\
\midrule
0.1 & 1.8\% & 0.2\% & 1.8\%\\
0.25 & 4.5\% & 0.6\% & 4.6\%\\
0.5 & 8.8\% & 1.3\% & 8.9\%\\
1.0 & 16.5\% & 3.5\% & 16.9\%\\
1.5 & 22.6\% & 6.2\% & 23.5\%\\
2.0 & 26.9\% & 9.2\% & 28.5\%\\
3.0 & 30.7\% & 14.6\% & 34.1\%\\
5.0 & 29.0\% & 20.6\% & 35.7\%\\
10.0 & 30.0\% & 22.6\% & 34.0\%\\
\bottomrule
\end{tabular}
\caption{Displacement discrepancy \eqref{eq:D} between the capsule operator
$\psi_A$, the sixteen-drop operator $\psi_D$ of Sec.~\ref{sec:transfer}, and
quasi-static injection of the same area by a uniform line source in 2-D
potential flow, $\psi_L$, as a function of aspect ratio. Line strokes of the
recovered programs have $\ell/r$ quartiles 1.0, 1.9, 3.6 (90th percentile
6.0). At the median aspect ratio $\psi_D$ is closer to the potential flow than
$\psi_A$ is. All three agree as $\ell\to0$: at $\ell=0$ the capsule is
identically the incompressible cylinder-insertion map
$\sigma=\sqrt{\sigma'^2-r^2}$, which is where every insertion operator must
start.}\label{tab:foliation}
\end{table}

\begin{figure*}[t]\centering
\noindent\makebox[0.25\textwidth]{\footnotesize Capsule ($\ell=3r$)}%
\makebox[0.25\textwidth]{\footnotesize Sixteen ordered drops}%
\makebox[0.25\textwidth]{\footnotesize One area-matched drop}%
\makebox[0.25\textwidth]{\footnotesize Line-source flow}\\[2pt]
\includegraphics[width=\textwidth,trim=0 0 0 26bp,clip]{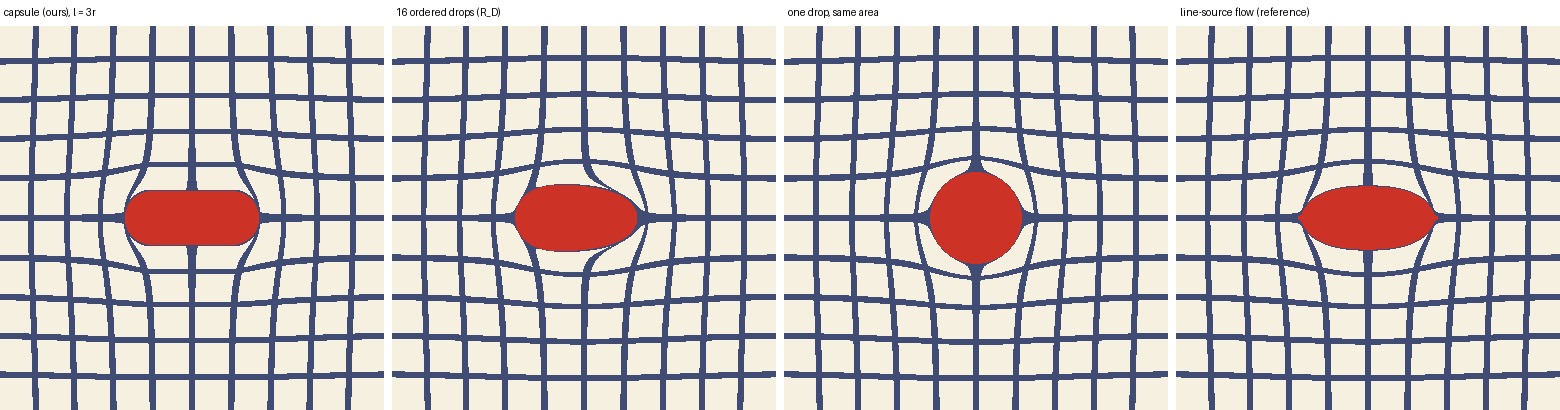}
\caption{\textbf{One stroke of the same inserted area under four
operators.} A grid carried by the sheet is pulled back through each backward
map; the deposit is red. Left to right: the capsule (ours, $\ell=3r$), one
closed-form map that inserts an elongated deposit; sixteen ordered drops
along the same segment, whose displacement is used by $\RD$ in
Sec.~\ref{sec:transfer}, with visible order effects;
a single drop of the same area \cite{Lu2012}, which cannot produce a drawn
deposit; and the line-source potential flow $\psi_L$, the reference
of Eq.~\eqref{eq:D}. The transfer experiment retains capsule coverage
while substituting the ordered-drop displacement.}\label{fig:prim}
\end{figure*}

\subsection{A variant closer to line-source potential flow}\label{sec:calibrated}
To mimic line-source potential flow more closely without giving up the
closed-form inverse, we keep the foliation and the coordinate $v$ and let the
increment in the area coordinate vary along a leaf. The map is then an
additive coupling layer in the area chart, in the sense of NICE
\cite{Dinh2015}: one coordinate is shifted by a function of the other, which
is left unchanged,
\begin{equation}
\Phi_g:\; u\longmapsto u+A_\ell(r)\,g(v),\qquad v\longmapsto v,
\label{eq:coupling}
\end{equation}
so the Jacobian is unit triangular and the inverse subtracts the same
quantity. With $g$ of unit mean the inserted area is still $A_\ell(r)$, and the
deposit becomes $\{u\le A_\ell(r)g(v)\}$ with an explicit boundary. Every
measure-preserving map that keeps $v$ has the form \eqref{eq:coupling}, so \eqref{eq:coupling} exhausts such variants, none of which can vary with
distance from the segment. We take
\begin{equation}
g(v)=1+\sum_{k=1}^{3}c_k\cos 4\pi k v,
\label{eq:profile}
\end{equation}
whose modes have zero mean and respect both reflections of
Proposition~\ref{prop:tan}. The first mode is $+1$ at the flank midpoints and
$-1$ at the tips, so a positive $c_1$ widens the flank and shortens the tips, the two ways in
which the reference deposit differs from the equal-area stadium.
Positivity of $g$ is checked exactly, since in $w=\cos4\pi v$ it is a cubic on
$[-1,1]$.

The coefficients are fitted once per aspect ratio to the forward line-source
map, on starting distances from $0.02r$ to $8r$ with held-out angles and
distances, and interpolated between the fitted ratios; the supplement
(Sec.~S9) lists them with the full comparison. At the median aspect ratio of
the recovered strokes the discrepancy \eqref{eq:D} falls from $27\%$ to $8\%$,
and it stays near $8\%$ across the aspect ratios that occur; the
deposit-boundary error falls from $0.17r$ to $0.03r$, at negligible extra cost per point. The residual is structural rather than a fitting limit:
the reference also moves material along the leaves, and a map that keeps $v$
cannot follow, which bounds \eqref{eq:D} below at $5.6\%$ at the median aspect
ratio; and because the increment cannot vary with distance, the far field is
slightly worse than under the uniform increment for short strokes. Fitting
programs under \eqref{eq:coupling} at the headline settings, with the
coefficients following each stroke's aspect ratio, retains image quality: over
five references and three seeds the hard-render error is $1.9\%$ lower in
geometric mean than the headline programs, inside the seed band, with
LPIPS-Alex level (supplement, Sec.~S9). Every other experiment below uses
$g\equiv1$.

\section{Inverse rendering}\label{sec:inverse}
\subsection{Parameterisation and relaxation}\label{sec:param}
Every constraint is enforced by construction, so \eqref{eq:inverse} is
unconstrained in $\theta$. Half-width is
$r=r_{\mathrm{lo}}+(r_{\mathrm{hi}}-r_{\mathrm{lo}})\,\sgm(\theta_r)$; offset is
$d=\ell_{\max}\tanh(\theta_d)$; anchor is
$a=a_{\mathrm{lo}}+(a_{\mathrm{hi}}-a_{\mathrm{lo}})\,\sgm(\theta_a)$
componentwise, with a window chosen either to keep the anchor on the canvas
or to contain the whole footprint; colour is
$c=\mathrm{softmax}(\theta_c/\tau_c)\,\mathsf{P}$ over a trainable palette
$\mathsf{P}\in[0,1]^{k\times3}$, itself a sigmoid of free parameters.
Table~\ref{tab:hyper} gives every constant.

The hard coverage $\mathbf{1}_{K_i}$ has no useful gradient. We relax it to
\begin{equation}
\alpha_\tau(x;g)=\sgm\!\Big(\frac{r-\sigma(x)}{\max(\tau r,\tau_{\min})}\Big),
\label{eq:relax}
\end{equation}
and anneal $\tau$ geometrically over the optimisation, so that early
iterations see a smooth objective and late ones a sharp one; because
$\tau_{\min}>0$, optimisation always uses a finite-width relaxation and never
the hard coverage itself. The relaxed coverage does not integrate to
$A_\ell(r)$, so during optimisation the model is only approximately
area-consistent; every reported number is computed at the hard render. Optimisation is by Adam with a cosine step-size
schedule; every stroke, the palette and every colour weight are optimised
jointly. Stroke order and count are fixed.

\paragraph{The deep-feature objective.} Three series --- the stroke-count
and stroke-type ablations of Sec.~\ref{sec:abl} and the paint-budget sweep of
Sec.~\ref{sec:process} --- were trained under a deep-feature data term rather
than pixel MSE, and are read within themselves. It is the unit-normalised
VGG16 feature distance: the reference and the render are normalised by the
ImageNet channel statistics and passed through the convolutional stack of
VGG16 with ImageNet weights; at \texttt{relu1\_2}, \texttt{relu2\_2},
\texttt{relu3\_3} and \texttt{relu4\_3} the per-location feature vectors of
both images are normalised to unit length across channels, their squared
difference is summed over channels and averaged over locations, and the four
layer terms are summed. This is LPIPS-shaped but carries no learned weights,
and is treated as a perceptual proxy in its own units.

\begin{table*}[t]\centering\footnotesize\setlength{\tabcolsep}{5pt}
\begin{tabular}{p{0.22\textwidth}p{0.72\textwidth}}\toprule
strokes $N$ / parameters & 2000 / 26{,}024 (13 per stroke, palette $8\times3$)\\
half-width window & $r\in[0.6\,r^0,\,1.7\,r^0]$; $r^0$ geometric from $0.120$ to $0.005$ over the program, so $r\in[0.003,0.204]$\\
offset & $d=\ell_{\max}\tanh\theta_d$, $\ell_{\max}=0.30$\\
coverage relaxation & $\tau$: $0.35\to0.02$ geometric; floor $\tau_{\min}=0.7/\text{res}$\\
colour & softmax over an $8$-colour palette, temperature $1.0\to0.05$\\
regulariser & $\lambda\sum_i\ell_i$, $\lambda=3\times10^{-4}$\\
optimiser & Adam; steps $0.012$ (anchor), $0.018$ (width), $0.015$ (offset),
$0.030$ (colour), $0.010$ (palette), with cosine decay; 240 iterations at $1024^2$\\
objective scale & data term scaled by $\mathcal{L}_{\text{VGG}}(0)/\mathcal{L}_{\text{MSE}}(0)$ at initialisation, per run\\
evaluation & hard render $\mathbf{1}_{K}$ at $1024^2$, reference implementation, seed 0\\
\bottomrule
\end{tabular}
\caption{Optimisation settings. Constants were selected on flower and battal, together with one further sheet
that is not part of the evaluation corpus, and then held fixed for the
remaining references and matched controls.}\label{tab:hyper}
\end{table*}

\paragraph{Objective and metrics.} The data term is pixel mean-squared error,
computed on sRGB values as stored, without a transfer-function decode: the
reference is read as $8$-bit sRGB and divided by $255$, and the palette lives
in the same space. Linear light would reweight the error towards
highlights; every arm and every comparator is scored under the same convention. We use LPIPS-Alex \cite{Zhang2018} only as a secondary perceptual
measurement. Optimisation uses the relaxed coverage of
Eq.~\eqref{eq:relax}; every reported image error uses an exported program
with hard coverage. The relaxation used for gradients therefore does not enter any reported score.

\subsection{The adjoint of the chain}\label{sec:adjoint}
Optimising a program of $N$ coupled strokes requires a gradient through all
$N$ of them, and keeping every intermediate state for the backward pass is infeasible at the chain lengths required here. Two properties of the operator replace it:
it is invertible in closed form, so states can be \emph{regenerated} instead
of stored, and its non-invertibility is confined to an explicitly characterisable set, so the exceptions can be \emph{recorded}. Figure~\ref{fig:adjoint} shows the
construction; Table~\ref{tab:adjoint} gives its cost.

\begin{figure}[t]\centering
\begin{tikzpicture}[font=\footnotesize,>={Latex[length=1.6mm]},
  st/.style={draw,rounded corners=1pt,minimum height=4.2mm,minimum width=7.5mm,
             inner sep=1pt,fill=black!3},
  gh/.style={draw,densely dotted,rounded corners=1pt,minimum height=4.2mm,
             minimum width=7.5mm,inner sep=1pt,text=black!45}]
% ---- forward lane
\node[anchor=base west,font=\footnotesize\bfseries] at (-0.15,1.30) {render (forward)};
\node[st] (fN) at (0.35,0.75) {$s_N$};
\node[gh] (fa) at (1.55,0.75) {$s_{i}$};
\node[gh] (fb) at (2.75,0.75) {$s_{i-1}$};
\node[st] (f0) at (4.15,0.75) {$s_0$};
\draw[->] (fN)--(fa); \draw[->,densely dotted] (fa)--(fb); \draw[->,densely dotted] (fb)--(f0);
\node[anchor=west,text=black!45,font=\scriptsize] at (4.62,0.75) {kept};
\node[anchor=west,text=black!45,font=\scriptsize] at (1.15,0.28)
     {discarded as produced};
% ---- backward lane
\node[anchor=base west,font=\footnotesize\bfseries] at (-0.15,-0.62) {adjoint (backward)};
\node[st] (b0) at (0.35,-1.15) {$s_0$};
\node[st] (ba) at (1.55,-1.15) {$s_{i-1}$};
\node[st] (bb) at (2.75,-1.15) {$s_{i}$};
\node[st] (bN) at (4.15,-1.15) {$s_N$};
\draw[->] (b0)--(ba); \draw[->] (ba)-- node[above,font=\scriptsize] {$\Phi_i$} (bb);
\draw[->] (bb)--(bN);
\node[anchor=west,font=\scriptsize] at (4.62,-1.15) {$O(P)$};
% ---- the obstruction
\node[draw,rounded corners=1pt,fill=red!6,draw=red!55,align=center,
      font=\scriptsize,inner sep=2.5pt] (tab) at (2.30,-2.35)
     {$\psi_i$ loses coordinates inside deposit $i$\\[1pt]
      side table supplies $(i,x_i)$ for $x_i\in K_i$};
\draw[->,red!65] (tab.north) -- (2.30,-1.38);
\end{tikzpicture}
\caption{\textbf{Regenerate, and trace the exceptions.} The render produces
per-pixel states $s_i=(x_i,C_i,T_i)$ in the order $N\to0$; reverse-mode
differentiation needs them in the order $0\to N$. Exterior positions are
reconstructed with $\Phi_i$, and the remaining state with
\eqref{eq:rT}--\eqref{eq:rC}. Positions lost inside deposits are read from a
side table (Lemma~\ref{lem:obstruction}). The diagram shows the $O(P)$ active
state; total memory also includes the side table and periodic position
checkpoints.}\label{fig:adjoint}

\end{figure}
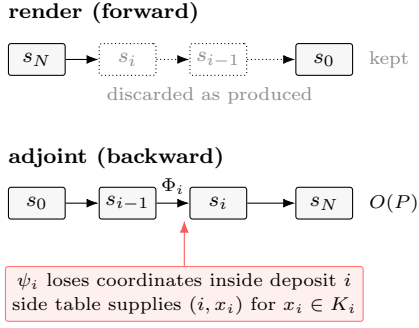

Write the per-pixel state after stroke $i$ has been processed as
$s_{i-1}=(x_{i-1},C_{i-1},T_{i-1})$ --- position, accumulated colour,
transmittance; indices decrease as processing proceeds, matching
\eqref{eq:chain} --- so that \eqref{eq:chain}--\eqref{eq:render} read
\begin{equation}
\begin{aligned}
C_{i-1}&=C_i+T_i\alpha_i c_i, & T_{i-1}&=T_i(1-\alpha_i),\\
x_{i-1}&=\psi_i(x_i),&&
\end{aligned}
\label{eq:state}
\end{equation}
with $s_N=(x,0,1)$ and $\mathcal{R}(\Pi)(x)=C_0+T_0c_0$. Reverse-mode
differentiation of $\mathcal{L}$ through \eqref{eq:state} requires the states
$s_i$ in the order $i=0,\dots,N$; a forward pass produces them in the opposite
order and retains only $s_0$. Storing all of them costs $O(NP)$ for $P$
pixels --- about $50$\,GB for $2000$ strokes at $1024^2$, an analytic estimate for a configuration that was not run --- and checkpointing \cite{Griewank2000} trades it for
recomputation. Instead we \emph{regenerate} them. Because $\psi_i$ is
invertible when restricted to the exterior (Proposition~\ref{prop:area}, with inverse
\eqref{eq:fwd}),
\begin{align}
x_i&=\Phi_i(x_{i-1}),\label{eq:rpos}\\
\log T_i&=\log T_{i-1}-\log(1-\alpha_i),\label{eq:rT}\\
C_i&=C_{i-1}-T_i\,\alpha_i\,c_i,\label{eq:rC}
\end{align}
which is \eqref{eq:state} solved for $s_i$. The backward pass steps
\eqref{eq:rpos}--\eqref{eq:rC} alongside the vector-Jacobian products of
\eqref{eq:state}, holding $O(P)$ state.

\begin{lemma}[Support of the replay record]\label{lem:obstruction}
Let $K_i$ be the closed footprint. The exterior map $\Phi_i$ is a bijection
from $\R^2\setminus\{\mathrm{segment}_i\}$ onto $\R^2\setminus K_i$.
Consequently, \eqref{eq:rpos} reconstructs $x_i$ uniquely whenever
$x_i\notin K_i$. For $x_i\in K_i$, the extended $\psi_i$ maps to the
segment, outside the domain of $\Phi_i$; exterior inversion cannot reconstruct
such samples. The exceptional set is the closed footprint $K_i$, which the implementation
records (below).
\end{lemma}
\begin{proof}
For $x_i\notin K_i$, Proposition~\ref{prop:area} gives
$x_i=\Phi_i(\psi_i(x_i))$. For $x_i\in K_i$, the clamped extension sets
$\sigma=0$, so $\psi_i(x_i)$ lies on the segment and \eqref{eq:rpos} is
undefined. The extension collapses two-dimensional regions to that segment
and is not injective. The boundary $\partial K_i$ has zero planar area.
\end{proof}

The ideal map is exactly invertible on the exterior. The implementation
regularises the root \eqref{eq:back}: it floors $C$ at $(\epsilon r_i)^2$ with
$\epsilon=10^{-3}$, which sends the footprint to a distance of at most
$\epsilon r_i/\sqrt{\pi}$ from the segment while retaining the direction, and it
records every sample at which the floor binds, $C<(\epsilon r_i)^2$, which is
the closed footprint together with an exterior sliver of width below
$10^{-6}r_i$. Boundary ties, on which collapsed trajectories concentrate, are
therefore stored rather than regenerated, and \eqref{eq:rpos} is applied only
to samples the floor did not touch, where it is exact up to float32 rounding;
the replay error of Table~\ref{tab:adjoint} measures the implemented scheme.

\paragraph{Relation to reversible backpropagation.} Existing approaches that avoid storing a long trajectory share the assumption that the forward step is invertible.
Reversible architectures are built so that it can \cite{Gomez2017};
reversible learning integrates the dynamics backwards \cite{Maclaurin2015}; path replay re-traces a light path from its
seed \cite{Vicini2021}. A painting operator breaks that assumption
differently: it \emph{creates} material, so the map is not injective on the deposition set and no integration scheme can recover the lost state. The applicable principle is that of reversible learning \cite{Maclaurin2015}, to store exactly what cannot be regenerated and nothing else, specialised here to a set that is identifiable during
the forward pass and whose size we measure below rather than bound. We state
it that way in Sec.~\ref{sec:transfers} because it is the part of this paper
most likely to be useful outside marbling.

The lemma gives the record this replay scheme needs: we store $(i,x_{i,p})$ for every pixel
trajectory $p$ that the floor touches, and regenerate every other position
from \eqref{eq:rpos}.

\begin{remark}[Size of the side table]
Let $S=\sum_{i,p}\mathbf{1}[x_{i,p}\in K_i]$ be the
implemented record count. A unit-density estimate would be
$P\sum_i|K_i|/|\Om|$. It is not a general bound: collapsed trajectories can
travel together and be counted again at subsequent insertions, and the pixel
grid introduces sampling error. At $N=2000$ we measured $9.9$--$14.0$
records per pixel at $128^2$, $256^2$ and $1024^2$ to the reported precision
($0.12$--$0.18$\,GB at $1024^2$). The corresponding unit-density area
estimates are $4.4$ for the initial program and $5.5$--$6.2$ for recovered
programs. These are measurements on the tested programs, not bounds.
\end{remark}

Soft coverage makes transmittance reconstruction ill-conditioned when
$\alpha_i$ approaches one, and a long float32 composition accumulates
position error even though each map preserves area. Both are conditioning problems, and both are handled by choosing where to
spend memory: transmittance is carried in logarithms with $\alpha_i$ clamped at $1-10^{-7}$, and positions are checkpointed periodically at an interval chosen from the sweep of Table~\ref{tab:adjoint}; the gradient error that remains at each interval is measured.
With checkpoint interval $k$ and the record count $S$ defined above,
\[
M=O\!\left(P+P\lceil N/k\rceil+S\right).
\]
The constants depend on the implementation; $S$ is the measured record
count of the remark above.

\paragraph{Structure of the gradient.} Write $\lambda^C_{i-1}$,
$\lambda^T_{i-1}$ and $\lambda^x_{i-1}$ for the adjoints of $\mathcal{L}$
with respect to the state $s_{i-1}$ that step $i$ produces. Differentiating
\eqref{eq:state} gives the adjoint step
\begin{equation}
\begin{aligned}
\lambda^C_i&=\lambda^C_{i-1},\qquad
\lambda^T_i=(1-\alpha_i)\,\lambda^T_{i-1}+\alpha_i\,\lambda^C_{i-1}\!\cdot c_i,\\
\lambda^x_i&=(D\psi_i)^{\!\top}\lambda^x_{i-1}+\mu_i\,\partial_x\alpha_i,\\
\mu_i&=T_i\big(\lambda^C_{i-1}\!\cdot c_i-\lambda^T_{i-1}\big),
\end{aligned}
\label{eq:adjstep}
\end{equation}
and the parameter gradients
\begin{equation}
\frac{\partial\mathcal{L}}{\partial c_i}=\sum_x\lambda^C_{i-1}\,T_i\,\alpha_i,\qquad
\frac{\partial\mathcal{L}}{\partial\theta_i}=\sum_x\Big[\mu_i\,\frac{\partial\alpha_i}{\partial\theta_i}
+\lambda^x_{i-1}\!\cdot\frac{\partial\psi_i}{\partial\theta_i}\Big]
\label{eq:adjparam}
\end{equation}
for the geometric parameters $\theta_i\in\{a_i,d_i,r_i\}$. Three regions of
the canvas contribute differently (Fig.~\ref{fig:holegrad}); with the relaxed
coverage \eqref{eq:relax} they are regions where one term dominates, not exact
supports. Inside the footprint --- the \emph{hole} --- $\alpha_i$ is close to
one: the stroke's colour receives most of its gradient there, weighted by the
transmittance $T_i$ left by the strokes laid after it, so paint that is later
covered receives little colour gradient. In a band of width
$\max(\tau r_i,\tau_{\min})$ around $\partial K_i$ --- the \emph{rim} ---
$\partial_\theta\alpha_i$ is largest: this is where the shape of the deposit is
adjusted, through the sensitivity $\mu_i$ of the composite to coverage. Outside
--- the \emph{rest} --- $\alpha_i$ is small and the transport term dominates:
$\partial_\theta\psi_i$ is non-zero on the whole exterior and falls off with
distance from the footprint, so changing the position, width or length of stroke $i$ changes where
every earlier deposit is sampled, and $\lambda^x_{i-1}$ carries the sensitivity
of all that earlier paint. The term $(D\psi_i)^{\!\top}\lambda^x_{i-1}$ pulls that accumulated
sensitivity back through stroke $i$'s map into the frame in which strokes
$i+1,\dots,N$ act, so each later stroke's transport term receives gradient contributions from pigment deposited before stroke $i$. Inside the hole the extended $\psi_i$ collapses to the segment
and its derivative is rank-deficient, but there $T_{i-1}=T_i(1-\alpha_i)$ is
small, so the colour gradient of earlier paint is attenuated at the same pixels where position has little effect on the rendered image. Replay makes
every quantity in \eqref{eq:adjstep}--\eqref{eq:adjparam} available at step
$i$ from $s_{i-1}$ and the side table; the capsule supplies $\psi_i$, its
inverse and their derivatives in closed form, so no per-sample iterative
solve is needed (an iterative alternative is compared in the supplement).

\begin{figure}[t]\centering
\begin{tikzpicture}[font=\scriptsize,>={Latex[length=1.4mm]}]
% rest: transport arrows on the exterior, decaying outward
\foreach \ang in {15,45,75,105,135,165,195,225,255,285,315,345}{
  \pgfmathsetmacro{\cx}{ifthenelse(cos(\ang)>0.05,2.35,ifthenelse(cos(\ang)<-0.05,1.15,1.75))}
  \draw[->,black!55,thick] ({\cx+0.95*cos(\ang)},{0.6+0.95*sin(\ang)}) -- ({\cx+1.45*cos(\ang)},{0.6+1.45*sin(\ang)});}
% rim (relaxed coverage band, width tau)
\fill[orange!30] (1.15,-0.2) -- (2.35,-0.2) arc (-90:90:0.8) -- (1.15,1.4) arc (90:270:0.8) -- cycle;
% hole (footprint)
\fill[red!45] (1.15,0.05) -- (2.35,0.05) arc (-90:90:0.55) -- (1.15,1.15) arc (90:270:0.55) -- cycle;
\draw[red!70!black] (1.15,0.05) -- (2.35,0.05) arc (-90:90:0.55) -- (1.15,1.15) arc (90:270:0.55) -- cycle;
\draw[dashed,orange!80!black] (1.15,-0.2) -- (2.35,-0.2) arc (-90:90:0.8) -- (1.15,1.4) arc (90:270:0.8) -- cycle;
\node at (1.75,0.6) {hole};
\node[anchor=west] at (3.55,0.95) {rim, width $\max(\tau r,\tau_{\min})$};
\node[anchor=west] at (3.55,0.25) {rest};
\draw[orange!80!black] (3.5,0.95) -- (3.0,1.28);
\draw[black!55] (3.5,0.25) -- (3.35,-0.35);
\end{tikzpicture}\\[3pt]
\begin{tabular}{@{}l@{\hskip 6pt}l@{}}\scriptsize
\textbf{hole} $\alpha_i\to1$ & \scriptsize colour gradient $\lambda^C_{i-1}T_i\alpha_i$, weighted by $T_i$\\
\scriptsize\textbf{rim}, relaxation band & \scriptsize shape gradient $\mu_i\,\partial_\theta\alpha_i$\\
\scriptsize\textbf{rest} $\alpha_i=0$ & \scriptsize transport gradient $\lambda^x_{i-1}\!\cdot\partial_\theta\psi_i$, falling off with distance\\
\end{tabular}
\caption{\textbf{Where one stroke's gradient comes from.} Inside its
footprint the stroke receives a colour gradient, weighted by the
transmittance of the paint laid above it; on the relaxed rim it receives a
shape gradient through the coverage; on the whole exterior it receives a
transport gradient through $\psi_i$, because resizing or moving it changes
where every earlier deposit is sampled. Paint beneath the hole is attenuated ($T$ small) and receives little of either gradient
there; the spatial sensitivity of the earlier paint that remains visible is
pulled back through $(D\psi_i)^{\!\top}$ to the strokes laid after it.}\label{fig:holegrad}
\end{figure}
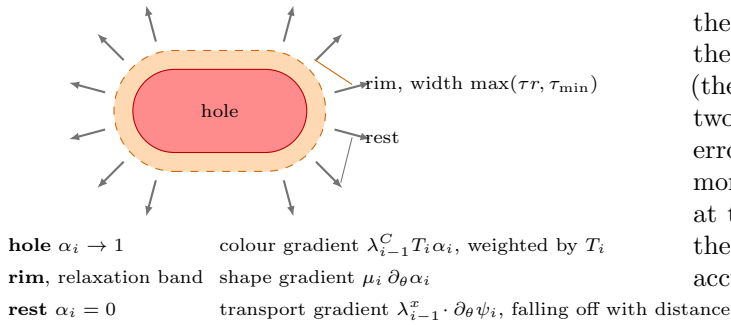

\begin{table}[t]\centering\footnotesize\setlength{\tabcolsep}{3pt}
\begin{tabular}{l rrr}\toprule
backward pass & peak (GB) & step (s) & vs.\ AD\\
\midrule
stored trajectory & $\approx50$ & --- & ---\\
autograd, checkpoints every 20 & 6.49 & 23.00 & 1.0$\times$\\
replay adjoint (PyTorch, ckpt 256) & 0.75 & 23.89 & 8.7$\times$\\
\textbf{replay adjoint, fused (Triton)} & 1.72 & 0.52 & 3.8$\times$\\
\bottomrule
\end{tabular}

\vspace{5pt}
\begin{tabular}{l rr rr}\toprule
checkpoint & peak (GB) & vs.\ AD & grad.\ err & cosine\\
\midrule
64 & 1.32 & 4.9$\times$ & $4.2\times10^{-3}$ & 0.9999914\\
128 & 0.91 & 7.1$\times$ & $5.5\times10^{-3}$ & 0.9999849\\
256 & 0.71 & 9.1$\times$ & $7.4\times10^{-3}$ & 0.9999731\\
512 & 0.61 & 10.6$\times$ & $1.1\times10^{-2}$ & 0.9999408\\
none & 0.54 & 12.1$\times$ & $1.7\times10^{-2}$ & 0.9998566\\
\bottomrule
\end{tabular}
\caption{The adjoint at the operating point ($N=2000$, $1024^2$, float32, one
24\,GB GPU). Top: storing the trajectory at this depth would take about
$50$\,GB for the $N\!\times\!P$ per-pixel states, an analytic figure we did
not run, so the practical baseline is automatic differentiation with
checkpoints every $20$ strokes. That interval was fixed a priori and not swept; the card had headroom.
Regenerating the trajectory instead costs $0.75$\,GB at a
comparable step time ($23.89$ against $23.00$\,s); the fused kernels trade some of that back for a $44\times$
faster step than checkpointed autograd, and are used for the transport-arm
fits. The side table holds
9.88 records per pixel (124\,MB). Bottom: the checkpoint interval trades memory
against float32 drift in the regenerated positions; we use $256$. The
sweep is a separate run from the top table, which is why its $256$ row
reads $0.71$ rather than $0.75$\,GB. Gradient
error is relative to the checkpointed-autograd gradient, exact up to
float32 rounding, and the cosine similarity of the
full parameter gradient never falls below $0.9998$.}\label{tab:adjoint}
\end{table}

\paragraph{Memory and runtime.} Table~\ref{tab:adjoint}. Storing the full per-pixel
trajectory at $N=2000$, $1024^2$ would need about $50$\,GB for the state
alone, an analytic figure. Against the practical baseline --- automatic
differentiation with checkpoints every $20$ strokes, $6.49$\,GB ---
regenerating costs $0.75$\,GB at a comparable step time ($23.89$ against
$23.00$\,s). The baseline interval ($20$) was fixed before the comparison and was not swept; across our own sweep the ratio ranges from $4.9$ to $12.1\times$, and it is $8.7\times$ at the adopted interval of $256$ in the run of the top block
(the sweep's own run at that interval gives $9.1\times$; the two are separate runs, as the caption states). Gradient error against the
checkpointed-autograd gradient grows monotonically with the interval and is
$7.4\times10^{-3}$ relative at the adopted one, with cosine similarity
$0.99997$ over the full parameter vector, consistent with float32 error
accumulating in the regenerated positions.

\paragraph{Fitting performance.} The absolute footprint is the practical figure: the fused implementation
peaks at $1.72$\,GB and the PyTorch
recurrence at $0.75$\,GB for a $2000$-stroke program at $1024^2$, well within one consumer GPU. A full optimisation of $240$ steps takes
$234$\,s on one RTX PRO 4000. The optimisation step --- forward, backward and
update --- takes $0.55$\,s at $N=2000$ and scales linearly in $N$: $0.075$,
$0.15$, $0.30$, $0.55$ and $1.0$\,s at $N=250$, $500$, $1000$, $2000$ and
$4000$; at $512^2$ it is $0.075$\,s and at $2048^2$ $2.5$\,s.

\paragraph{Implementation.} The vector--Jacobian products of
\eqref{eq:state} are written analytically rather than obtained by differentiating the renderer,
so that the whole step can be fused into a single kernel; this is the
difference between the PyTorch and fused rows in the upper block of
Table~\ref{tab:adjoint}. They agree with automatic differentiation to
$5.4\times10^{-15}$ in float64 and with central differences of the rendered
value to $2.5\times10^{-9}$; the float32 error of the full $2000$-stroke
replay is the $7.4\times10^{-3}$ of Table~\ref{tab:adjoint}.

\section{Experiments}\label{sec:exp}
\subsection{Protocol}\label{sec:design}
We use photographs of five \emph{ebru} marbled sheets as target images: flower,
tulip, teardrops, battal, and lale. \emph{Ebru} is a living
practice, inscribed on the UNESCO Representative List of the Intangible
Cultural Heritage of Humanity in 2014, and it has its own named, ordered
stroke vocabulary: \emph{battal} is the sprinkled ground, and the floral
patterns are drawn with a stylus over such a ground \cite{Wolfe1990}. That
vocabulary is not ours. The pattern types of four of the five references are
traditionally made largely by zero-injection stylus and comb work, which our
stroke family excludes (Sec.~\ref{sec:limits}); the recovered programs are
therefore not the traditional strokes. This also bears on
Sec.~\ref{sec:groundtruth}: where the model's stroke family differs from the
one that made a sheet, there is no correspondence to recover. All five are photographs of marbled sheets published on Wikimedia Commons
under free licences, cropped to a centre square fixed before any fitting and
resampled to $512^2$; the supplement lists each file, its author and its
licence, and the three out-of-domain artworks of Sec.~\ref{sec:result} are
public-domain reproductions from the same source. The crops are preprocessed to $512^2$
and are resampled to $1024^2$ by Lanczos filtering for
fitting and scoring; the $512^2$ column of Table~\ref{tab:resolution} shows
the effect of replay and scoring at that common resolution. The images are used only
as optimisation targets. Flower and battal were used while selecting the constants in
Table~\ref{tab:hyper}, together with a sixth sheet excluded from the corpus because its source could not be documented; the remaining three were held out from that process. Every main result is a hard render of an
exported program. Unless noted otherwise, each program has $N=2000$ strokes
and $26{,}024$ optimised parameters, and is fitted for 240 Adam steps from the
same rule-based initial program.

We denote by $\Pi_A$ the program obtained after the finite-step optimisation
of Eq.~\eqref{eq:inverse} with the relaxed transport renderer. For the
structural control, $\Pi_U$ is obtained with each material map set to the
identity, using the same parameter family and step budget. 

$\Pi_A$ is trained through the fused implementation and $\Pi_U$ through the
reference implementation, which alone exposes the inert-mark path. The
data-term scale of Table~\ref{tab:hyper} is set from each arm's own initial
loss and differs between arms by $1.36$--$2.22\times$, so the effective
regularisation differs by that factor. That factor does not account for the
difference the control reports: $\Pi_U$ reaches a \emph{lower} loss than
$\Pi_A$ does under each arm's own renderer on 4/5 references in MSE and 5/5 in
LPIPS-Alex (Sec.~\ref{sec:result}), which weighs against under-optimisation of the control as the explanation. Adam is also approximately invariant to a global scale on
the objective, so the mismatch acts mainly on the relative weight of the length
regulariser. To keep
evaluation matched, both exported programs are re-rendered by the same hard
reference implementation under both models, producing the four cells
$\mathcal{L}(\mathcal{R}_M(\Pi_X),I^\ast)$, $M,X\in\{A,U\}$: two optimisation
runs per reference, each scored twice. Pixel MSE is primary and LPIPS-Alex is
reported as a secondary metric.

\paragraph{External stroke fitter.} Stylized Neural Painting \cite{Zou2021} fits
parametric brush strokes to an image through a differentiable stroke renderer. Its
strokes are inert marks, so it is a raster comparator: it can be compared with
$\Pi_A$ only as an image and with $\Pi_U$ only under $\mathcal{R}_U$. We use the
authors' code and pre-trained oil-paint renderer with the released defaults: the
progressive schedule of five grid levels with $36$ strokes per block, giving $1980$
strokes of $12$ parameters ($23{,}760$ values, against our $26{,}024$), of which
$1756$--$1902$ pass the renderer's minimum-size test; the paper's objective ($L_1$ plus
optimal transport); and a hard render of the final strokes at $1024^2$. Each reference
is fitted from a white and from a black initial canvas and the better is kept. The
render is scored against the same $1024^2$ target by the same code as every other
row. Each fit takes about 7\,min on the same card, with two initial canvases per
reference. A second seed
moves each reference's error by $0.7$--$3.7\%$.

\paragraph{Error scale and controls.} Reference images differ in variance by
almost an order of magnitude, so every error we report is normalised by the
MSE of the best constant image, which puts the five references on one scale.
With one exception the comparators are internal ablations: the
inert-mark arm shares this paper's vocabulary, initial program, seed,
resolution and step budget and differs in $\psi_i$ and, through the data-term
scale, in effective regularisation (above). The exception is a published stroke fitter run at matched stroke count, which places that control
against a system we did not build. Where two
arms both sit above the constant-image baseline, we report their ratio only
alongside it. Three series --- the
ablations of Sec.~\ref{sec:abl} and the paint-budget sweep of
Sec.~\ref{sec:process} --- were trained under the deep-feature objective and
are read against their own noise band in those units. The spreads in
Table~\ref{tab:noise} measure how far a re-run moves and are used as such.

\subsection{Recovered programs}\label{sec:result}
Figure~\ref{fig:recovery} shows the primary result: a single target image is
converted into a replayable 2000-stroke program under the prescribed
deposit-and-transport model; Table~\ref{tab:2x2} lists the errors. Repeated
motifs and large silhouettes are retained, while the all-over battal pattern
remains the hardest reference because its smallest structures approach the
stroke-width floor. Battal is the one pure-insertion pattern in the set and
the model fits it worst, while the stylus-drawn florals, whose making our
stroke family does not model at all, score best. The floor ratio therefore
follows the image statistics rather than fidelity to the medium; the
out-of-domain targets below show the same.

The scale for these errors is the best constant image, whose MSE is the
image variance. Every pixel-MSE image error in the cross-method comparisons is reported normalised by that floor. The recovered programs sit below it on all five references, by a geometric
mean of $3.67\times$ (Table~\ref{tab:2x2}), and re-fitting from three seeds
moves that figure by less than the seed band (Sec.~\ref{sec:stats}). The ratio serves as a scale for the errors and should not be read as a
ranking of methods: a $64^2$ thumbnail with half the
parameters reaches a lower pixel error still
(Table~\ref{tab:compensation}), as any transport-free representation can.

\paragraph{An external stroke fitter.} Stylized Neural Painting, run at
matched stroke count (Sec.~\ref{sec:design}), is a raster of the same order as
the recovered program: below the flat floor on 5/5 (geometric mean
3.11$\times$, interval $[2.35, 4.13]$), ahead of $\Pi_A$
in pixel MSE on 2/5 references and behind it on 3/5, and behind it
in LPIPS-Alex on 4/5 (Table~\ref{tab:2x2}; the supplement gives
the LPIPS cells and the renders). Its role here is to place the matched
control. Under its own renderer $\mathcal{R}_U$, the inert-mark arm $\Pi_U$ is
a raster fitter of the same standing as the published one --- ahead of it in
pixel MSE on 4/5 and in LPIPS-Alex on 5/5 --- so the
program that lands above the floor once executed was not produced by a weak
stroke fitter; it was produced by a competent one and executed under the
transport its fitting ignored. That execution is the comparison a stroke-based
method cannot enter, since its output has no transport arm.

The matched control, which differs from the transport arm in $\psi_i$ and in data-term scale (Sec.~\ref{sec:design}), shows the role of coupled recovery. Fitting the same
vocabulary with every material map set to the identity and then executing the
result gives $\Pi_U$, which is $6.5\times$ worse than $\Pi_A$ and \emph{above}
the flat floor on $5/5$ references (Table~\ref{tab:2x2}): executed under
transport, a program fitted without it no longer reconstructs the target.

Under $\mathcal{R}_U$, $\Pi_U$ reaches a lower loss than $\Pi_A$ does under
$\mathcal{R}_A$ on most references: fitting without transport is the easier
problem, and if a
raster approximation is all that is wanted, modelling transport is a cost.
The benefit of modelling transport is a program that remains valid when executed under the transport model. The held-out references, which took no part in choosing the constants, score
no worse than the two development references (Table~\ref{tab:2x2}).

\paragraph{Out of domain.} The floor ratio measures image approximation and nothing specific to
marbling: the same solver with the same settings, given three artworks instead of marbled sheets, reaches
$3.8\times$ on \emph{The Great Wave}, $2.7\times$ on \emph{The Starry Night}
and $8.5\times$ on the \emph{Mona Lisa} (Fig.~\ref{fig:ood}), the last
higher than any ebru reference. The representation approximates non-marbling targets as readily; the three
calibrate the ratios in Table~\ref{tab:2x2}. The stroke fitter, run on the
same three, is ahead in pixel MSE on two of them (supplement).

\begin{figure*}[t]\centering
\noindent\makebox[0.3333\textwidth]{\footnotesize The Great Wave}%
\makebox[0.3333\textwidth]{\footnotesize The Starry Night}%
\makebox[0.3333\textwidth]{\footnotesize Mona Lisa}\\[2pt]
\noindent\makebox[0.1666\textwidth]{\scriptsize Target}%
\makebox[0.1666\textwidth]{\scriptsize Replay ($3.8\times$)}%
\makebox[0.1666\textwidth]{\scriptsize Target}%
\makebox[0.1666\textwidth]{\scriptsize Replay ($2.7\times$)}%
\makebox[0.1666\textwidth]{\scriptsize Target}%
\makebox[0.1666\textwidth]{\scriptsize Replay ($8.5\times$)}\\[2pt]
\includegraphics[width=\textwidth,trim=0 0 0 44bp,clip]{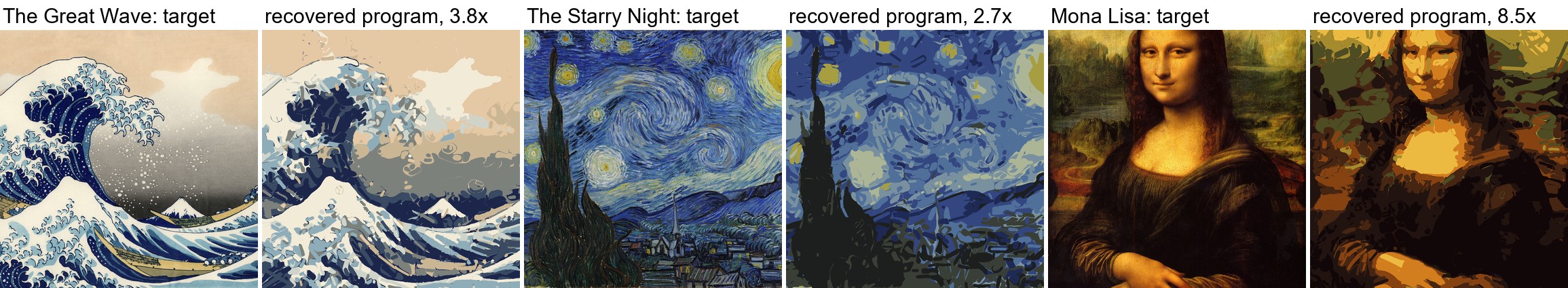}
\caption{\textbf{Out-of-domain targets.} Three artworks and the replay of
the 2000-stroke program recovered from each under the headline settings,
with the floor ratio. The programs approximate non-marbling targets
using the same deposition vocabulary; this ratio does not measure marbling fidelity.}\label{fig:ood}
\end{figure*}

\begin{table*}[t]\centering\footnotesize\setlength{\tabcolsep}{3.5pt}
\begin{tabular}{l r rr rr rr rr}\toprule
& & \multicolumn{4}{c}{transport applied ($\mathcal{R}_A$)}
& \multicolumn{2}{c}{transport omitted ($\mathcal{R}_U$)} & \multicolumn{2}{c}{external, raster}\\
\cmidrule(lr){3-6}\cmidrule(lr){7-8}\cmidrule(lr){9-10}
reference & flat floor & $\Pi_A$ & $\Pi_U$ & floor$/\Pi_A$ & $\Pi_U/$floor
& $\Pi_A$ & $\Pi_U$ & SNP & floor$/$SNP\\
\midrule
flower & 0.0151 & 0.0046 & 0.0282 & 3.25$\times$ & 1.87$\times$ & 0.0179 & 0.0053 & 0.0039 & 3.83$\times$\\
tulip$^\dagger$ & 0.0393 & 0.0135 & 0.0643 & 2.90$\times$ & 1.64$\times$ & 0.0399 & 0.0121 & 0.0152 & 2.58$\times$\\
teardrops$^\dagger$ & 0.0317 & 0.0061 & 0.0544 & 5.22$\times$ & 1.72$\times$ & 0.0353 & 0.0048 & 0.0077 & 4.14$\times$\\
battal & 0.0865 & 0.0334 & 0.1531 & 2.59$\times$ & 1.77$\times$ & 0.1285 & 0.0213 & 0.0329 & 2.63$\times$\\
lale$^\dagger$ & 0.0138 & 0.0026 & 0.0259 & 5.22$\times$ & 1.88$\times$ & 0.0154 & 0.0022 & 0.0051 & 2.72$\times$\\
\midrule
geometric mean & & & & \textbf{3.67$\times$} & \textbf{1.77$\times$} & & & & 3.11$\times$\\
95\% CI & & & & [2.43, 5.54] & [1.65, 1.91] & & & & [2.35, 4.13]\\
three seeds, GM / worst & & & & 3.71$\times$ / 3.66$\times$ & & & & & \\
\bottomrule
\end{tabular}
\caption{Recovery error against the flat floor, and the matched transport
control (hard-render pixel MSE at $1024^2$; lower is better). \emph{Flat floor}
is the MSE of the best constant image, the cheapest non-trivial reconstruction
of that reference. $\Pi_A$ is recovered jointly through the transport chain;
$\Pi_U$ is fitted with every material map set to the identity. Both exported
programs are scored by the same hard reference renderer under both models.
$\Pi_A$ beats the floor on 5/5 references (geometric mean $3.67\times$ below
it). $\Pi_U$, executed, is \emph{above} the floor on 5/5 --- fitting without
transport and then executing with it is worse than painting a single colour. Intervals are paired $t$-intervals on log-ratios. $^\dagger$: held out
during development. \emph{SNP} is Stylized Neural Painting \cite{Zou2021} at matched
stroke count, scored on its own $1024^2$ render (Sec.~\ref{sec:design}); it fits inert
marks and has no transport arm, so it is a raster comparator only.}
\label{tab:2x2}
\end{table*}

\begin{figure*}[t]\centering
\includegraphics[width=\textwidth]{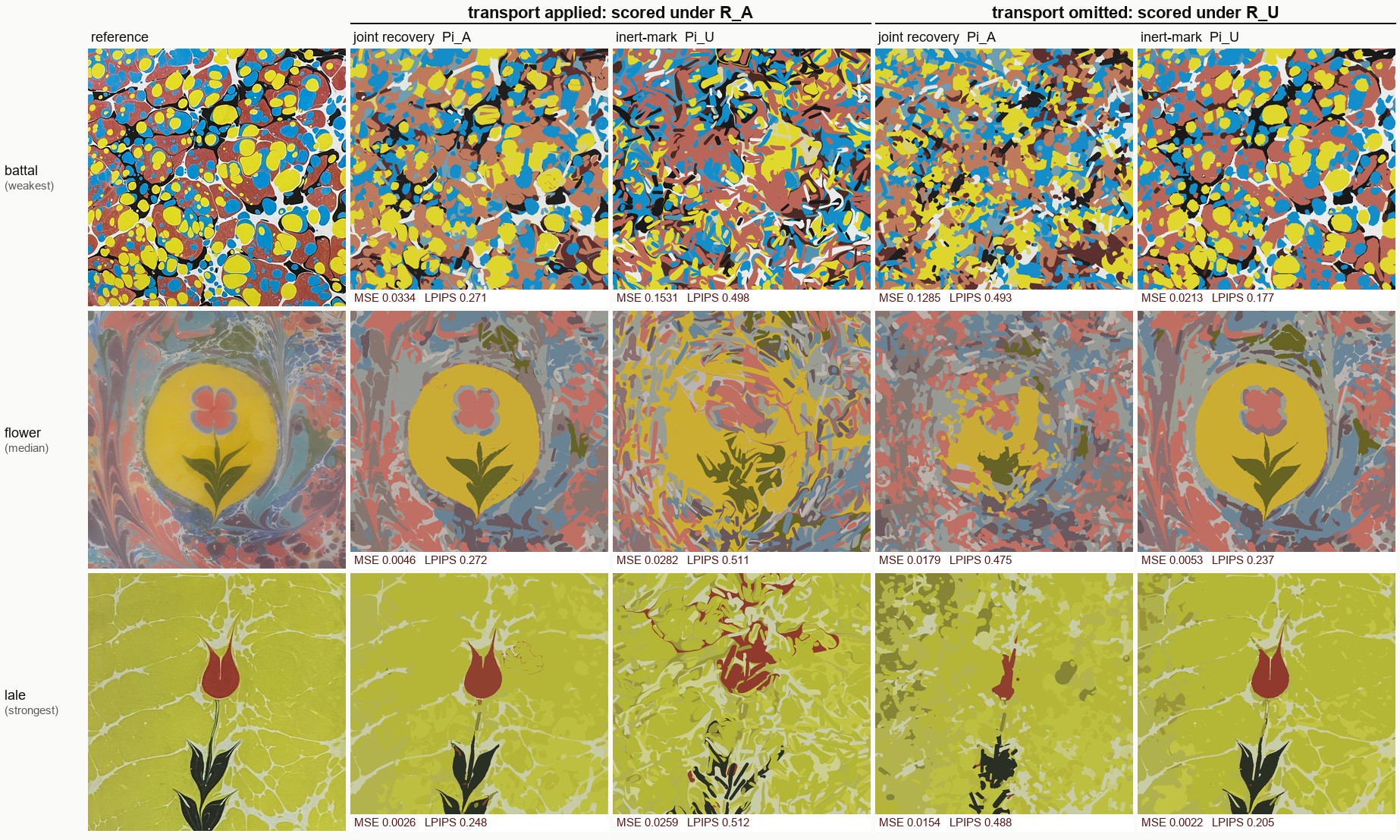}
\caption{\textbf{Why the ratios in Table~\ref{tab:2x2} are reported against a
floor.} Three of the five references --- battal, flower and lale: the weakest,
the median and one of the two strongest floor ratios. Left pair: both programs executed with
transport ($\RA$). Right pair: both scored with transport omitted ($\RU$).
Each program reconstructs under the model it was fitted for and disintegrates
under the other, but the failures are not degraded reconstructions: the two
inner result columns ($\Pi_U$ executed, $\Pi_A$ imagined) are visibly further
from the target than a single flat colour would be, as the floor column of
Table~\ref{tab:2x2} records. All
five references in the supplement.}\label{fig:hero}
\end{figure*}

\subsection{Closed-form compensation baseline}\label{sec:compensation}
A natural alternative to joint optimisation is to fit an inert-mark program
and then translate each stroke so that subsequent transport carries it to the position selected by the inert-mark fit. Because $\psi$ has a
closed-form inverse this can be done in one pass, with no gradients and no
re-optimisation. We evaluate two variants: \emph{anchor compensation}
pre-images the segment start and carries the offset vector unchanged, and
\emph{endpoint compensation} pre-images both ends, which is stronger although
it changes segment length and hence inserted area. Both run from the last
stroke to the first, so each stroke is pre-imaged through the
already-compensated geometry of every later one, at a cost of $N(N-1)/2$ point transforms for the anchor variant and twice that
for the endpoint variant ($1{,}999{,}000$ and $3{,}998{,}000$ at $N=2000$), and
no derivatives. Radii are left
unchanged: the deformation of a footprint by later strokes is exactly what a
one-pass point correction cannot repair.

\begin{table}[t]\centering\scriptsize\setlength{\tabcolsep}{2.5pt}
\begin{tabular}{l rrrr r}\toprule
reference & flat floor & inert replay & anchor & endpoint & blocked\\
\midrule
flower & 0.0151 & 0.0282 & 0.0175 & 0.0156 & 1466\\
tulip & 0.0393 & 0.0643 & 0.0480 & 0.0412 & 1454\\
teardrops & 0.0317 & 0.0544 & 0.0353 & 0.0291 & 1371\\
battal & 0.0865 & 0.1531 & 0.0974 & 0.0833 & 1362\\
lale & 0.0138 & 0.0259 & 0.0196 & 0.0148 & 1422\\
\bottomrule
\end{tabular}

\vspace{5pt}
\begin{tabular}{lrr}\toprule
method (geometric mean, five references) & GM hard MSE & vs.\ joint\\\midrule
flat floor (best constant image) & 0.029504 & $3.67\times$\\
$64^2$ thumbnail (12,288 values) & 0.006777 & $0.84\times$\\
SNP \cite{Zou2021}, 1980 strokes (23,760 values) & 0.009478 & $1.18\times$\\
inert replay ($\Pi_U$ under $\mathcal{R}_A$) & 0.052304 & $6.50\times$\\
anchor pre-image & 0.035544 & $4.42\times$\\
endpoint pre-image & 0.029715 & $3.70\times$\\
endpoint, blocked steps skipped & 0.029374 & $3.65\times$\\
\textbf{joint recovery (ours)} & \textbf{0.008041} & \textbf{$1.00\times$}\\
\bottomrule
\end{tabular}
\caption{One-pass closed-form compensation against joint recovery, all five
references at $1024^2$, hard render, one reference renderer. Each variant
pre-images an inert-mark program's strokes through every later stroke's
inverse map, without gradients or re-optimisation. The $64^2$ thumbnail is the reference downsampled and bilinearly
upsampled; it is not a method, and it is listed because it bounds what
pixel error alone can establish at this parameter count. The stroke-fitter row is the external raster comparator of Table~\ref{tab:2x2}, included on the same footing. \emph{Blocked} counts
strokes of the endpoint variant whose pre-image fails to exist at some step: $1362$--$1466$ of $2000$
($68$--$73\%$). By Lemma~\ref{lem:obstruction} no exact pre-image exists for those
strokes, so no correction of this form can place them, and they carry essentially all of the
point round-trip residual --- where a pre-image exists the correction is exact to
$7.9\times10^{-5}$ canvas units. Skipping the blocked steps instead of
clamping them changes the result by less than the re-run band
($3.65\times$ against $3.70\times$), so the baseline is not handicapped by
that choice. Endpoint compensation removes a median $44.6\%$ of the
inert-replay error and closes a median $52.4\%$ of the gap to joint recovery,
yet still lands at the trivial floor: it beats a single constant colour on only
$2$ of $5$ references. Its error is $3.70\times$ that of joint recovery
(95\% interval $[2.44, 5.60]$, $5/5$).}\label{tab:compensation}
\end{table}

Table~\ref{tab:compensation}. Closed-form geometry alone removes a median
$44.6\%$ of the inert-replay error and closes a median $52.4\%$ of the gap to
joint recovery. But the compensated program still only reaches the flat floor: it beats a
single constant colour on 2 of 5 references. One-pass correction remains less accurate than joint fitting on these examples.

Lemma~\ref{lem:obstruction} identifies an obstruction to this pointwise correction. Pre-imaging stroke $i$ through a later stroke $j$ requires a
point of $\Phi_j^{-1}$; if the point lies inside $K_j$ there is none, because
the extended $\psi_j$ maps all of $K_j$ to the segment, on which $\Phi_j$ is
undefined. This occurs for $1362$--$1466$ of the $2000$ strokes. Where a pre-image does
exist the correction is essentially exact --- round trip within
$7.9\times10^{-5}$ canvas units --- and where it does not, the error is
$0.02$--$0.04$, several hundred times larger. The large point round-trip residuals occur in
the blocked subset; the diagnostic is geometric, so it locates the failure of
the correction rather than partitioning the image-space residual. The same loss of coordinates that replay storage handles obstructs an exact
pre-image correction.

Two checks establish that the comparison is not an artefact of its construction: skipping a blocked
step instead of clamping it onto the segment moves the geometric mean from
$3.70\times$ to $3.65\times$, inside the re-run band, and compensating both
endpoints rather than the anchor alone is better ($3.70\times$ against
$4.42\times$), so the ceiling is not an artefact of correcting too few points.
Closed-form knowledge of the transport map is therefore not sufficient; joint recovery also adjusts width, colour and visibility, so the comparison measures joint fitting, not gradients alone.

\subsection{Transfer across transport constructions}\label{sec:transfer}
To test whether recovery overfits the foliation of our capsule map, we score
both programs under a third renderer $\RD$ that was not used for fitting.
Stroke $i$'s displacement is replaced by a composition of $K=16$ published
drop maps \cite{Lu2012} placed along its segment, each of radius
$s=\sqrt{A_{\ell_i}(r_i)/(K\pi)}$ so that the inserted area equals the
stroke's own. Coverage is unchanged; only $\psi_i$ is replaced. Each drop map preserves
area on its exterior. Because capsule coverage is retained, $\RD$ is a
controlled hybrid renderer, not a simulation of the complete coloured
ordered-drop deposition process. Its displacement is built from a published
primitive and differs from $\RA$: under \eqref{eq:D} it is $17$--$36\%$ from $\psi_A$ over the first-quartile to
ninetieth-percentile aspect ratios of the recovered strokes
(Table~\ref{tab:foliation}), and it is \emph{closer} to line-source
potential flow than $\psi_A$ is ($9\%$ against $27\%$ at the median aspect ratio). Figure~\ref{fig:prim} shows the two
operators, a single drop of the same area, and the reference flow acting on
a grid.

The substitution degrades $\Pi_A$, and Table~\ref{tab:transfer} quantifies the effect: executed under the alternative operator, $\Pi_A$ is worse by $17.8$ to
$51.6\%$ than under the operator it was fitted for (median $24.8\%$), while $\Pi_U$, whose fit is invalidated under either transport renderer, scores within $1\%$ of its $\RA$
value; against its own fitting renderer $\RU$, which applies no transport,
its change is several hundred percent, and the table reports the $\RA$
comparison for both programs. That degradation measures how much of the
fit is specific to our foliation. What survives it is that $\Pi_A$ remains a reconstruction ---
still below the flat floor on 5/5 --- while $\Pi_U$ remains above it, and the
$4.9\times$ separation holds (95\% interval $3.6$--$6.8$, same direction on
all five references and in LPIPS-Alex).

$\RD$ is a finite ordered-drop alternative built from a published primitive, not an independent physical model: its $K$ drops are
laid in one sweep along the segment, so it carries an injection schedule of
its own, and refining $K$ is not a limit process towards the line source.
Table~\ref{tab:foliation} measures what makes it a useful arbiter: at the
median aspect ratio it sits closer to $\psi_L$ than our map does. The
reference $\psi_L$ is not itself the evaluator for reasons of cost: each point is
carried by an RK4 integration of the source velocity, which in our implementation
is too expensive to compose $2000$ times at image resolution, whereas a
closed-form displacement composes at the cost of a few arithmetic operations. The arbiter is therefore built from a published closed-form primitive. The test varies the foliation of an insertion operator at
matched inserted area; it says nothing about a different displacement class,
such as the zero-injection tine and stylus operators of \cite{Lu2012}, which
are area-preserving and invertible and would be the natural next arbiter.

\begin{table*}[t]\centering\footnotesize\setlength{\tabcolsep}{6pt}
\begin{tabular}{l rr rr r}\toprule
& \multicolumn{2}{c}{$\Pi_A$} & \multicolumn{2}{c}{$\Pi_U$} & \\
\cmidrule(lr){2-3}\cmidrule(lr){4-5}
reference & MSE & vs.\ $\mathcal{R}_A$ & MSE & vs.\ $\mathcal{R}_A$ & $\Pi_U/\Pi_A$\\
\midrule
flower & 0.0055 & +17.8\% & 0.0282 & -0.1\% & 5.1$\times$\\
tulip & 0.0165 & +22.3\% & 0.0647 & +0.6\% & 3.9$\times$\\
teardrops & 0.0088 & +45.5\% & 0.0544 & -0.0\% & 6.2$\times$\\
battal & 0.0417 & +24.8\% & 0.1526 & -0.3\% & 3.7$\times$\\
lale & 0.0040 & +51.6\% & 0.0257 & -1.0\% & 6.4$\times$\\
\midrule
median & & +24.8\% & & -0.1\% & 5.1$\times$\\
GM [95\% CI] & & & & & 4.93$\times$ [3.58, 6.78]\\
\bottomrule
\end{tabular}
\caption{Transfer to $\mathcal{R}_D$, a displacement built from sixteen
area-matched drop maps per stroke, which neither program was fitted against.
\emph{vs.\ $\mathcal{R}_A$} is the change against the same program's score
under our transport renderer (Table~\ref{tab:2x2}), for both programs:
$\Pi_A$, fitted under $\mathcal{R}_A$, pays $+17.8$ to $+51.6\%$ for the
foreign operator, median $+24.8\%$; $\Pi_U$, which either transport
destroys, moves $-1.0$ to $+0.6\%$ (against its own fitting renderer
$\mathcal{R}_U$, which applies no transport, its change is several hundred
percent). $\Pi_A$
remains ahead on 5/5 in MSE and 5/5 in LPIPS-Alex.}
\label{tab:transfer}
\end{table*}

\begin{figure}[!htb]\centering
\includegraphics[width=\columnwidth]{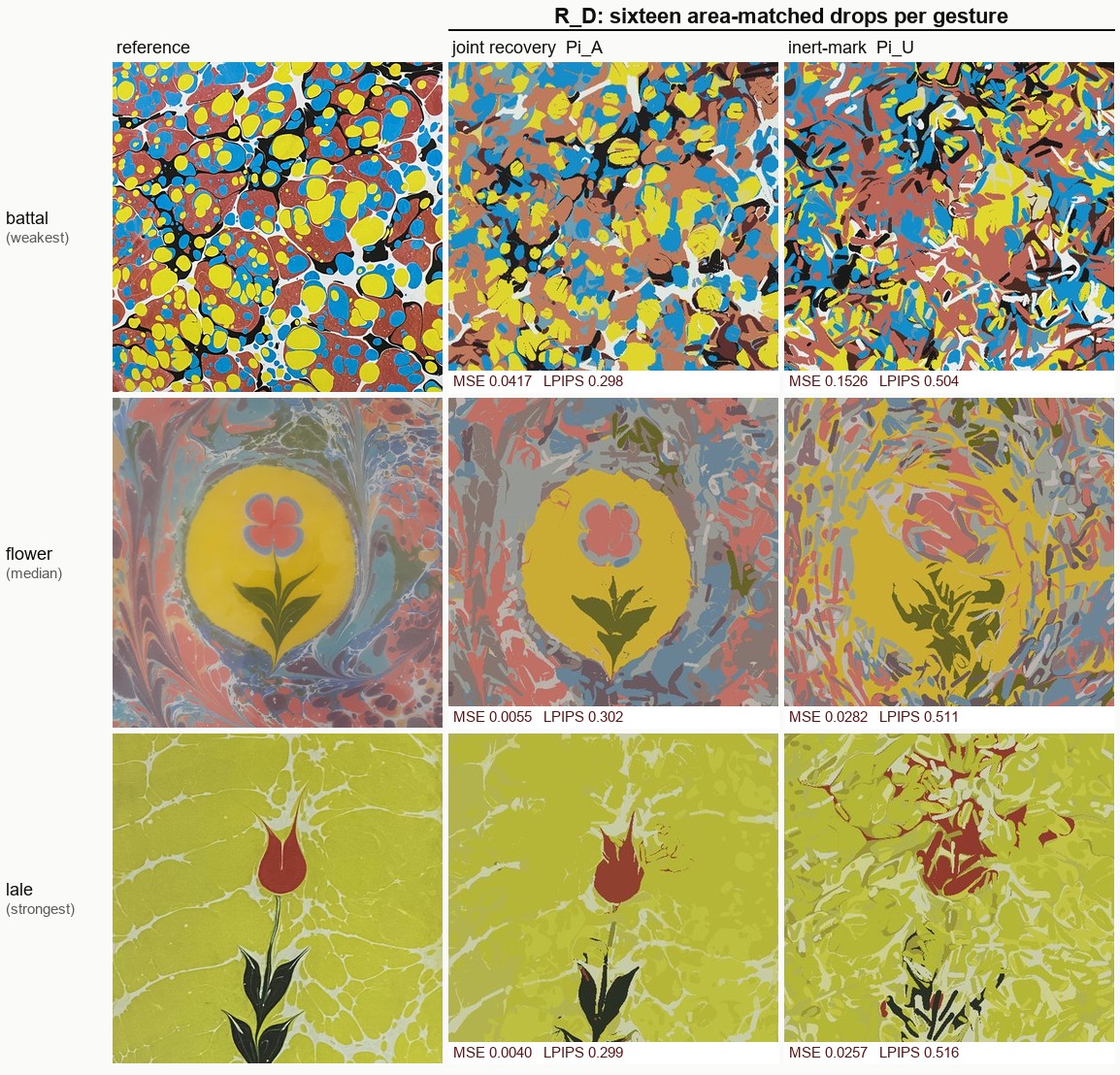}
\caption{\textbf{Under a transport neither program was fitted against.} Both
programs executed under $\RD$, sixteen area-matched drop maps per stroke,
same three references as Fig.~\ref{fig:hero}. $\Pi_A$ degrades --- by a median
$+24.8\%$ (Table~\ref{tab:transfer}) --- but remains a reconstruction;
$\Pi_U$ does not become one. All five in the supplement.}\label{fig:transfer}
\end{figure}

\subsection{Synthetic ground truth}\label{sec:groundtruth}
Every result so far is measured against historical sheets, where we cannot
know whether \emph{any} 2000-stroke program reproduces the target. That
confounds two very different failures: insufficient model expressivity, and failure of the optimiser to find a good fit. To separate them we build targets with known generating programs. Twelve programs are drawn from the model's own
prior --- the solver's radius schedule, the empirical offset-length
distribution of the recovered programs (median $\ell=0.018$, $78\%$ above the
export threshold, matching Sec.~\ref{sec:abl}), uniform anchors and a random
eight-colour palette --- and rendered under $\RA$. Each target is therefore
the exact render of a program that lies inside the feasible set, so a
zero-error solution provably exists. We then recover from the image alone,
with the settings of Table~\ref{tab:hyper} unchanged.

\begin{figure*}[t]\centering
\includegraphics[width=\textwidth]{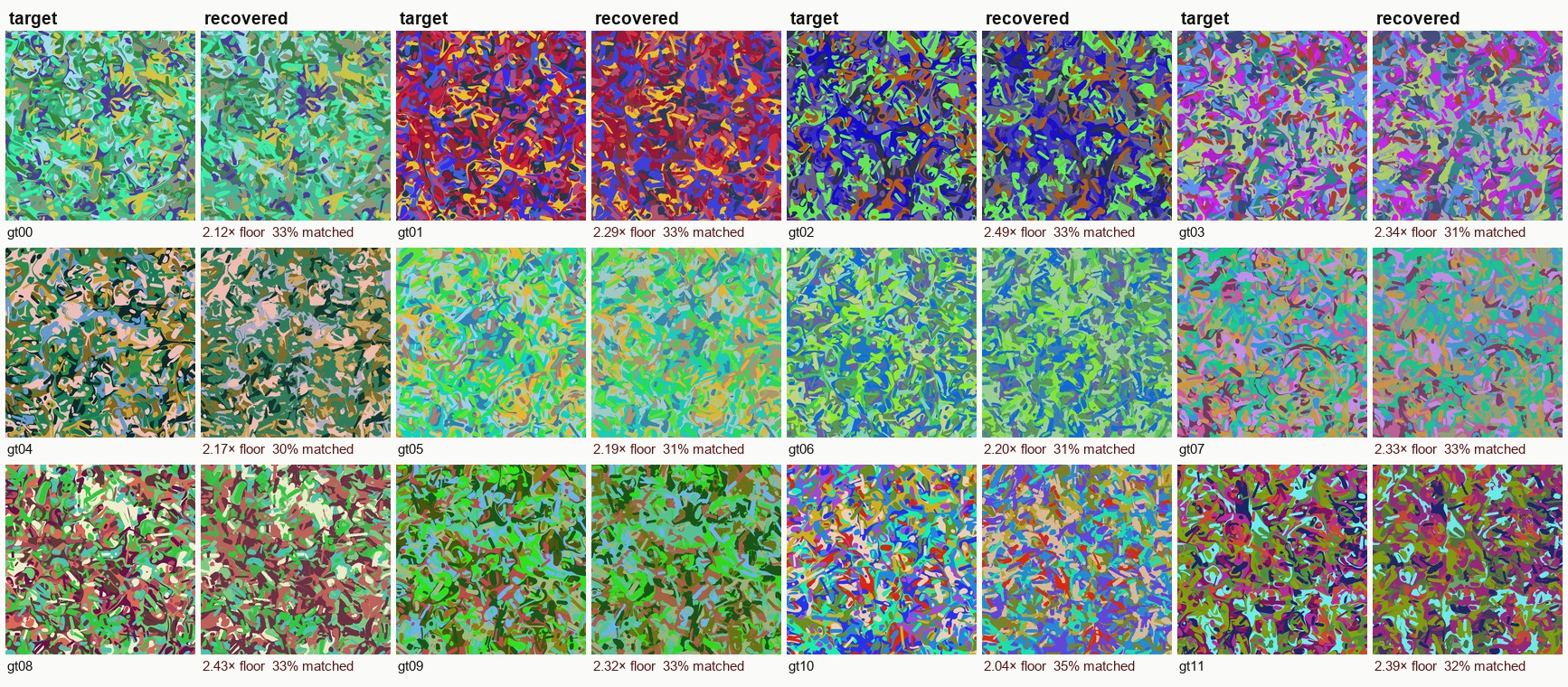}
\caption{\textbf{Synthetic ground truth.} Each target is the render of a known
2000-stroke program drawn from the model's prior; beside it is the render of
the program recovered from that image alone. Annotations give the floor ratio
and the fraction of visible strokes matched inside the true radius, from
the matching pass that produced the figure; the final matched-null analysis
of Table~\ref{tab:a2} differs from these annotations by up to one point.
These targets deliberately do not look like marbling:
the prior places anchors uniformly and draws colours at random, so the images
carry none of the large-scale structure a real sheet has. This experiment
tests the fitting procedure, not the medium.}\label{fig:a2}
\end{figure*}

\begin{table*}[t]\centering\footnotesize\setlength{\tabcolsep}{5pt}
\begin{tabular}{l rrr rr rr rr}\toprule
& \multicolumn{3}{c}{image ($1024^2$)} & \multicolumn{2}{c}{visible strokes}
& \multicolumn{2}{c}{median match (canvas)} & \multicolumn{2}{c}{within true radius}\\
\cmidrule(lr){2-4}\cmidrule(lr){5-6}\cmidrule(lr){7-8}\cmidrule(lr){9-10}
target & recovered & floor & floor$/$rec & truth & recovered
& recovered & matched null & recovered & matched null\\
\midrule
gt00 & 0.0219 & 0.0464 & 2.12$\times$ & 1029 & 1206 & 0.0162 & 0.0161 & 33\% & 33\%\\
gt01 & 0.0314 & 0.0720 & 2.29$\times$ & 1071 & 1220 & 0.0170 & 0.0165 & 32\% & 33\%\\
gt02 & 0.0270 & 0.0672 & 2.49$\times$ & 1040 & 1250 & 0.0163 & 0.0166 & 32\% & 32\%\\
gt03 & 0.0249 & 0.0582 & 2.34$\times$ & 1020 & 1215 & 0.0169 & 0.0176 & 30\% & 31\%\\
gt04 & 0.0331 & 0.0718 & 2.17$\times$ & 1042 & 1222 & 0.0171 & 0.0160 & 30\% & 33\%\\
gt05 & 0.0268 & 0.0586 & 2.19$\times$ & 1046 & 1231 & 0.0171 & 0.0174 & 30\% & 34\%\\
gt06 & 0.0207 & 0.0455 & 2.20$\times$ & 1032 & 1215 & 0.0169 & 0.0173 & 31\% & 31\%\\
gt07 & 0.0239 & 0.0557 & 2.33$\times$ & 1018 & 1201 & 0.0168 & 0.0178 & 33\% & 29\%\\
gt08 & 0.0270 & 0.0655 & 2.43$\times$ & 1057 & 1239 & 0.0168 & 0.0170 & 33\% & 32\%\\
gt09 & 0.0259 & 0.0601 & 2.32$\times$ & 1015 & 1237 & 0.0166 & 0.0156 & 34\% & 36\%\\
gt10 & 0.0487 & 0.0991 & 2.04$\times$ & 1037 & 1242 & 0.0161 & 0.0171 & 35\% & 32\%\\
gt11 & 0.0286 & 0.0683 & 2.39$\times$ & 1045 & 1227 & 0.0166 & 0.0172 & 31\% & 32\%\\
\midrule
summary & & & \textbf{2.27$\times$} & & & 0.0168 & 0.0170 & 32\% & 32\%\\
\bottomrule
\end{tabular}
\caption{Synthetic ground truth. Each target is the render of a \emph{known}
2000-stroke program drawn from the model's prior, so a zero-error solution
provably lies inside the feasible set. \emph{Image}: recovery reaches a
geometric mean $2.27\times$ below the flat floor, against $3.67\times$ on the
historical references; on this family the residual is the fitting
procedure's gap, not model misfit.
\emph{Visible strokes} counts those owning at least $50$ pixels of the final
$1024^2$ image (scaled for the $512^2$ matching pass). Strokes below
this threshold are excluded by the matching protocol, not proved unidentifiable.
\emph{Median match} is the optimal-assignment distance between visible
recovered and visible true strokes; the \emph{matched null} repeats it with
another target's recovered program, giving a similarly sized visible
candidate pool produced by the same fitting procedure. The twelve targets form a paired design and we test it: the paired mean
difference in median matched distance is $-1.6\times10^{-4}$ canvas units
($t(11)=-0.77$, 95\% interval $[-6.1,+2.9]\times10^{-4}$, recovery nominally
better on 8 of 12), and the within-radius fractions differ on 6 targets each
way. The interval excludes any systematic reduction larger than $3.8\%$ of the
$0.0168$ median, which is the sensitivity of this test at $n=12$. Recovered
programs also carry $18\%$ more visible strokes than the truth on 12 of 12
targets. This positional comparison bounds one statistic under one recovery
procedure; it does not establish intrinsic non-identifiability.}\label{tab:a2}
\end{table*}

Figure~\ref{fig:a2} shows what recovery looks like here, and
Table~\ref{tab:a2} reports the image error and the stroke correspondence separately.

\paragraph{The image.} Recovery reaches a geometric mean $2.27\times$ below
the flat floor, against $3.67\times$ on the historical references. Exact representability does not make a target easier to fit. One candidate mechanism
is that the historical sheets carry large-scale structure the optimiser can
exploit while a program drawn from an uninformative prior does not
(Fig.~\ref{fig:a2}); testing it is future work. A zero-error program exists
here by construction, so on this family the residual is not model misfit but the optimiser's approximation gap, arising from the relaxed objective, the regulariser and the finite step budget. This bound applies to the synthetic family; the historical
sheets are addressed by the next control, which re-fits the reference renders of
the recovered historical programs. These targets are also representable by
construction. Recovery reaches a geometric mean $3.98\times$ below each render's own
constant-image baseline , compared with $3.67\times$ for the
original scans (one run per reference, headline settings). The targets and
their baseline variances differ, so these ratios do not bound model misfit on
the original scans. They show that the procedure leaves residual error even
when fitting images produced by its own model.

\paragraph{The program.} We restrict matching to strokes owning at least
$50$ pixels at $1024^2$ (scaled for the $512^2$ matching pass), about $1030$
of the $2000$ true strokes. Excluded strokes may still be partially visible and localisable; the threshold only defines the evaluated subset. We use optimal assignment on segment-midpoint distance and compare
with another target's recovered program, a null that shares the fitting
procedure and has a similarly sized candidate pool, unlike a null drawn from
another ground-truth program, which has systematically fewer visible
strokes.

The aggregate statistics are similar: median matched distance $0.0168$
against $0.0170$ for the null, and $32\%$ of matches inside the true radius
for both (Table~\ref{tab:a2}). The test therefore gives no evidence of generating-stroke correspondence beyond the cross-target null, for this procedure and this positional statistic.

\subsection{Aggregation and numerical variation}\label{sec:stats}
The reference image is the unit of analysis ($n=5$). Because losses are
positive and comparisons are multiplicative, we aggregate ratios on the log
scale and report geometric means with paired 95\% intervals; per-reference
values remain visible in every table. Table~\ref{tab:noise} gives the noise
floors, and there are two, in different units. In pixel MSE on the transport
arm --- the units of every cross-method comparison --- two runs of one
configuration and seed differ by $0.6$--$4.6\%$, and three seeds on all five
references spread $2.7$--$6.7\%$, median $6.1\%$. In the deep-feature units of the
ablation series of Sec.~\ref{sec:abl}, the across-seed spread is
$0.5$--$0.6\%$, and that is the band those series are read against. The seed runs use the headline configuration; seed $0$ differs from the
headline run by up to $4.6\%$, which we count in the spread. All internal programs are scored by one reference renderer from exported
programs, and every image comparison, including the external fitter's, uses
the same scoring code;
the exported programs score $0.9$--$14.6\%$ higher (worse) than the training
log's best-step score of the same program, which is taken before palette colours
are snapped at export (Table~\ref{tab:noise}); log scores are not used for any
comparison. The small reference set permits the large, consistent within-image comparisons reported here, but not population-level claims about marbling styles.

\begin{table*}[t]\centering\footnotesize\setlength{\tabcolsep}{6pt}
\begin{tabular}{l r c r r r}\toprule
contrast (ratio) & GM & 95\% CI & $t(4)$ & $d_z$ & sign\\
\midrule
executed, $\Pi_U/\Pi_A$ (MSE) & 6.50$\times$ & [4.19, 10.09] & 11.8 & 5.3 & 5/5\\
imagined, $\Pi_A/\Pi_U$ (MSE) & 5.12$\times$ & [3.14, 8.35] & 9.3 & 4.1 & 5/5\\
interaction, ratio of ratios (MSE) & 33.29$\times$ & [14.35, 77.23] & 11.6 & 5.2 & 5/5\\
transfer $\mathcal{R}_D$, $\Pi_U/\Pi_A$ (MSE) & 4.93$\times$ & [3.58, 6.78] & 13.9 & 6.2 & 5/5\\
executed, $\Pi_U/\Pi_A$ (LPIPS) & 1.93$\times$ & [1.80, 2.08] & 25.2 & 11.3 & 5/5\\
imagined, $\Pi_A/\Pi_U$ (LPIPS) & 2.37$\times$ & [2.05, 2.74] & 16.6 & 7.4 & 5/5\\
\bottomrule
\end{tabular}
\caption{Paired analysis on the log scale, $n=5$ references. GM is the
geometric-mean ratio, the interval a paired $t$-interval on the five
log-ratios, $d_z$ the standardised paired effect, and ``sign'' the number of
references with the predicted direction (one-sided sign-test floor
$2^{-5}=0.031$). The MSE contrasts between the inert arm and joint
recovery compare cells on opposite sides of the flat floor
(Table~\ref{tab:2x2}), so their size is not a quality statement. The minimum detectable paired effect at $n=5$, one-sided $\alpha=0.05$,
power $0.8$ is $d_z\approx1.33$. Executed MSE on the three
references not used during development: GM 7.5$\times$ [2.8, 20.2]
(development two: 5.3$\times$).}\label{tab:stats}
\end{table*}

\begin{table}[t]\centering\footnotesize\setlength{\tabcolsep}{2.5pt}
\begin{tabular}{l rr r r r}\toprule
& \multicolumn{3}{c}{re-run, same configuration} & \multicolumn{2}{c}{export comparison}\\
\cmidrule(lr){2-4}\cmidrule(lr){5-6}
reference & run 1 & run 2 & spread & log score & export vs.\ log\\
\midrule
flower & 0.004649 & 0.004596 & 1.1\% & 0.004609 & +0.9\%\\
tulip & 0.013519 & 0.013318 & 1.5\% & 0.013263 & +1.9\%\\
teardrops & 0.006070 & 0.006307 & 3.9\% & 0.005296 & +14.6\%\\
battal & 0.033431 & 0.033246 & 0.6\% & 0.032551 & +2.7\%\\
lale & 0.002633 & 0.002755 & 4.6\% & 0.002453 & +7.3\%\\
\bottomrule
\end{tabular}
\caption{Numerical variation, in the two units the paper uses. Left: two
optimisation runs of one configuration and seed, both scored by the hard
reference renderer, differ by 0.6--4.6\% in pixel MSE. Right: the training-time score of the best step, taken before palette
colours are snapped at export, against the reference render of the exported
program; the column is $100(\mathrm{MSE}_{\mathrm{export}}/
\mathrm{MSE}_{\mathrm{log}}-1)$. Rendering the exported program with both
kernels agrees within $0.15\%$, so the gap is colour snapping, not renderer
disagreement. Across three seeds on all five references the pixel-MSE spread is
2.7--6.7\%, median 6.1\%.
Seed $0$ differs from the headline run by up to $4.6\%$, which we count as
part of the spread. In the deep-feature
units of the ablation series the across-seed spread is 0.5--0.6\% (two
references, two seeds).}\label{tab:noise}
\end{table}

\begin{figure*}[!t]\centering
\includegraphics[width=\textwidth]{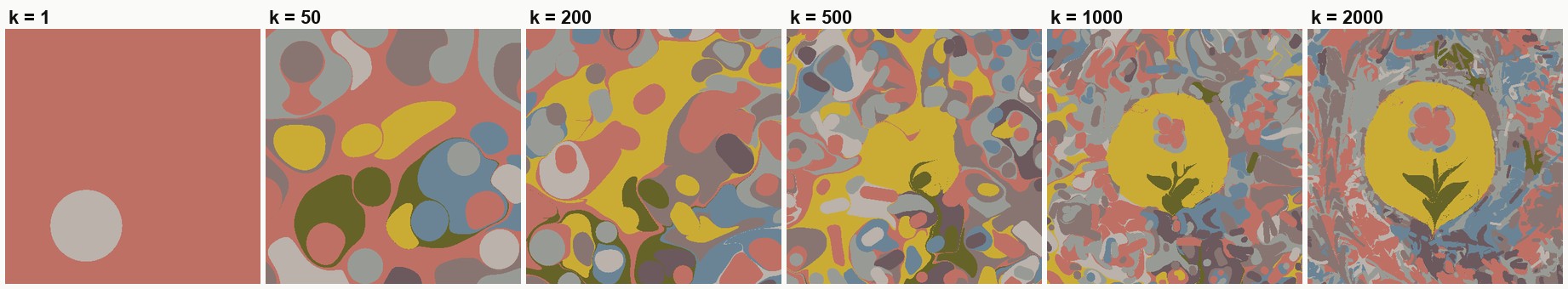}
\caption{\textbf{A recovered program can be replayed.} Hard renders of the
first $k$ strokes show the recovered flower program assembling in execution
order. This temporal decomposition is unavailable from a raster
reconstruction alone.}\label{fig:build}
\end{figure*}

\section{Program-level output}\label{sec:process}
The solver returns ordered parameters rather than a final raster image. Rendering
prefixes of the sequence exposes how the design is assembled
(Fig.~\ref{fig:build}); individual strokes, colours or whole stages can be
edited and replayed through the same dynamics.

\paragraph{Editing the program.} Figure~\ref{fig:interv} shows three representative edits, each replayed by the hard renderer under the same dynamics
with nothing re-optimised. Inserting one drop into the ground at stroke
$1000$ pushes the pigment laid before it and is itself pushed by everything
after it, so the vein pattern re-flows around it as the insertion model
prescribes while the motif is untouched. Recolouring the last fifth
of the program is exact, because transport does not depend on pigment, and
changes only the pixels those strokes own. Translating the last quarter
moves the late structure and re-flows the earlier pattern around it. The
edits change $7$--$41\%$ of the pixels. One attempted edit fails: selecting the strokes that
own a motif in the final image and translating their anchors does not
translate the motif, because each stroke's final position is the composition
of every later map, and the selection also captures early strokes that contribute only marginally to the region. Moving a motif therefore requires re-optimisation rather than a parameter
edit.

\begin{figure*}[t]\centering
\noindent\makebox[0.2\textwidth]{\footnotesize Target}%
\makebox[0.2\textwidth]{\footnotesize Recovered program}%
\makebox[0.2\textwidth]{\footnotesize Insert one drop}%
\makebox[0.2\textwidth]{\footnotesize Recolour last fifth}%
\makebox[0.2\textwidth]{\footnotesize Translate last quarter}\\[2pt]
\includegraphics[width=\textwidth,trim=0 0 0 44bp,clip]{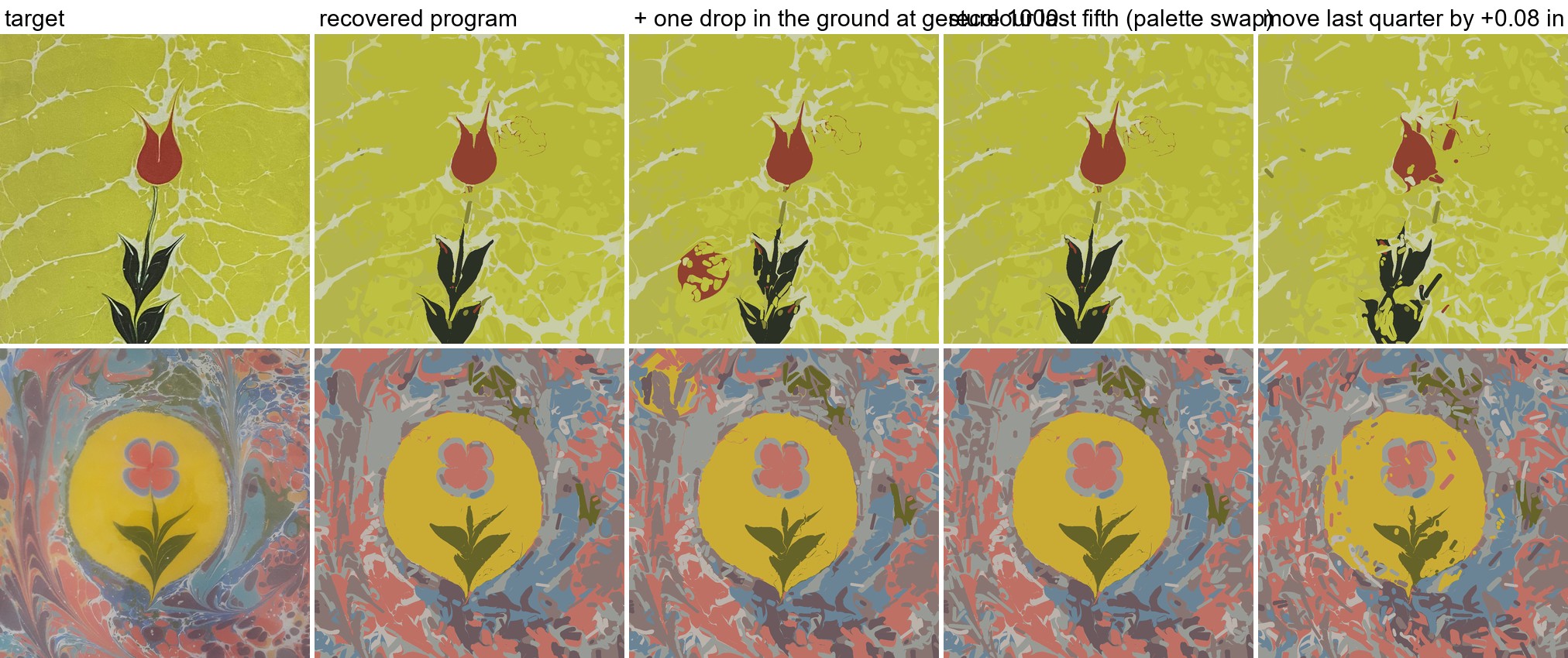}
\caption{\textbf{Program-level edits, replayed.} Lale and flower: target,
recovered program, and three edits of that program rendered by the same hard
renderer. Column three inserts a single drop of radius $0.06$ into the
ground beside the motif, at canvas position $(0.30,0.72)$ for lale and
$(0.20,0.20)$ for flower, in a palette colour taken from the motif, at
stroke $1000$ of $2000$; column four swaps the two most-used palette entries among the last
fifth of the strokes; column five translates the last quarter of the
strokes by $+0.08$ canvas units in $x$. Nothing is re-optimised.}\label{fig:interv}
\end{figure*}

\paragraph{Resolution.} Strokes are specified in canvas coordinates, so a
recovered program is a continuous-domain object and the raster resolution used for fitting should not be tied to it. Table~\ref{tab:resolution} tests that
directly: each program is optimised once at $1024^2$ and then replayed
unchanged at $512^2$ and $2048^2$, scored against the reference resampled to
the same grid. Replaying at $2048^2$ moves the error by at most $0.23\%$;
replaying at $512^2$ changes it by at most $4\%$. Above the $512^2$ preprocessing scale the target is an upsample, so those tests assess replay of the
representation, not recovery of new detail.

The converse does not hold for
fitting: programs fitted at $512^2$ and replayed at $1024^2$ reach
$3.13\times$ the floor against $3.67\times$ for programs fitted at $1024^2$
(per reference $-3$ to $-31\%$), at about a seventh of the optimisation time, because the
width floor and the relaxation are defined in pixels. A coarse-to-fine
schedule recovers most of that loss: fitting at $512^2$ for $200$ steps with
a free radius window and then continuing the same program at $1024^2$ for
$40$ steps, with the relaxation started near its floor, reaches $3.75\times$
in geometric mean against $3.67\times$ for the headline fit, in $37$\,s of
optimisation time against $132$\,s, one run per reference. The spread is wide --- flower improves to
$4.44\times$, lale falls from $5.22$ to $3.97\times$ because the coarse stage assigns a wrong colour to the flower head of the lale motif, which forty fine steps do not undo --- and relative to a $1024^2$ fit with the same free radius ($3.98\times$) the schedule loses $6\%$. Table~\ref{tab:speed} gives the times and the per-reference errors.

\begin{table*}[t]\centering\scriptsize\setlength{\tabcolsep}{3pt}
\begin{tabular}{l rr rrrrr r r}\toprule
& \multicolumn{2}{c}{time per fit (s)} & \multicolumn{5}{c}{hard-render MSE $\times10^{3}$ at $1024^2$} & GM & GM\\
\cmidrule(lr){2-3}\cmidrule(lr){4-8}
configuration & steps & wall & flower & tulip & teardrops & battal & lale & MSE$\times10^3$ & floor$/$MSE\\
\midrule
headline: $1024^2$, 240 steps, windowed radius & 132 & 234 & 4.65 & 13.52 & 6.07 & 33.43 & 2.63 & 8.04 & 3.67$\times$\\
free radius: $1024^2$, 240 steps & 132 & 234 & 4.69 & 12.81 & 5.31 & 29.95 & 2.34 & 7.41 & 3.98$\times$\\
$512^2$ only: 240 steps, windowed radius & 18 & --- & 5.74 & 15.19 & 6.55 & 34.32 & 3.82 & 9.44 & 3.13$\times$\\
stage A: $512^2$, 200 steps, free radius & 15 & 45 & 5.25 & 14.83 & 5.68 & 32.05 & 3.54 & 8.71 & 3.39$\times$\\
stage B: A, then 40 steps at $1024^2$ & 37 & 91 & 3.40 & 14.28 & 5.63 & 32.05 & 3.46 & 7.88 & 3.75$\times$\\
\bottomrule
\end{tabular}
\caption{Speed against error. \emph{Steps} is the optimisation time alone,
from the uncontended per-step measurements ($0.55$\,s at $1024^2$,
$0.075$\,s at $512^2$, $N=2000$). \emph{Wall} is a whole fit including start-up and the export render, which
differ by row; the step column is the fair comparison.
MSE is pixel mean-squared error of the exported program under the reference
renderer against the target at $1024^2$; \emph{floor$/$MSE} is the error of
the best single flat colour divided by the program's error, so $3.67\times$
means the program's error is $3.67$ times smaller than painting one colour.
One run per cell. The coarse-to-fine schedule (stage B) matches the headline's error in
$37$\,s of optimisation against $132$\,s, with a wide per-reference
spread: flower improves and lale loses, because the coarse stage commits to a
wrong colour for the flower head of the lale motif.}\label{tab:speed}
\end{table*}

\begin{table*}[t]\centering\footnotesize\setlength{\tabcolsep}{5pt}
\begin{tabular}{l rrr r rr rr r}\toprule
& \multicolumn{4}{c}{replayed at three resolutions (pixel MSE)}
& \multicolumn{2}{c}{PSNR (dB)} & \multicolumn{2}{c}{SSIM} & \\
\cmidrule(lr){2-5}\cmidrule(lr){6-7}\cmidrule(lr){8-9}
reference & $512^2$ & $1024^2$ & $2048^2$ & $1024\!\to\!2048$
& floor & $\Pi_A$ & floor & $\Pi_A$ & floor$/$MSE\\
\midrule
flower & 0.004663 & 0.004649 & 0.004651 & +0.03\% & 18.21 & 23.33 & 0.669 & 0.658 & 3.25$\times$\\
tulip & 0.013721 & 0.013519 & 0.013517 & -0.02\% & 14.06 & 18.69 & 0.377 & 0.490 & 2.90$\times$\\
teardrops & 0.006311 & 0.006070 & 0.006068 & -0.03\% & 14.99 & 22.17 & 0.488 & 0.539 & 5.22$\times$\\
battal & 0.034035 & 0.033431 & 0.033457 & +0.08\% & 10.63 & 14.76 & 0.269 & 0.400 & 2.59$\times$\\
lale & 0.002634 & 0.002633 & 0.002627 & -0.22\% & 18.62 & 25.80 & 0.714 & 0.751 & 5.22$\times$\\
\bottomrule
\end{tabular}
\caption{The recovered program is a continuous-domain object fitted to raster
samples. Each program is optimised once at $1024^2$ and then replayed
unchanged on a $512^2$, $1024^2$ and $2048^2$ grid, scored against the
reference resampled to the same grid. Replaying at $2048^2$ changes the error
by at most $0.23\%$ on any reference (median $0.04\%$); replaying at
$512^2$, the preprocessing scale, by at most $4\%$, comparable to the seed spread of Table~\ref{tab:noise}: the
resolution the program was fitted at is not baked into it. PSNR and SSIM are
reported alongside the same metrics for the best constant image as a
reference point: battal's $14.76$\,dB is
$+4.13$\,dB over its own floor. The PSNR column is a monotone re-expression of the same floor ratio
($+4.13$\,dB is $10\log_{10}2.59$), not an independent metric. In SSIM the recovered program exceeds the constant image on four references
and falls slightly below it on flower.}\label{tab:resolution}
\end{table*}

\paragraph{Spatial distribution of the residual.} Figure~\ref{fig:errmaps} shows target,
replay and amplified absolute error for the median and the hardest reference.
The residual is not spread evenly: it concentrates on the finest structures,
where the recovered stroke widths approach the width floor, and along the
boundaries of large colour regions. This is the same limit the stroke-count
series reaches (Sec.~\ref{sec:abl}) --- additional strokes capture finer structure until the improvement falls below the noise floor --- and it is why battal, an
all-over pattern of small features, is the reference the method reconstructs
least well.

\begin{figure}[!htb]\centering
\includegraphics[width=\columnwidth]{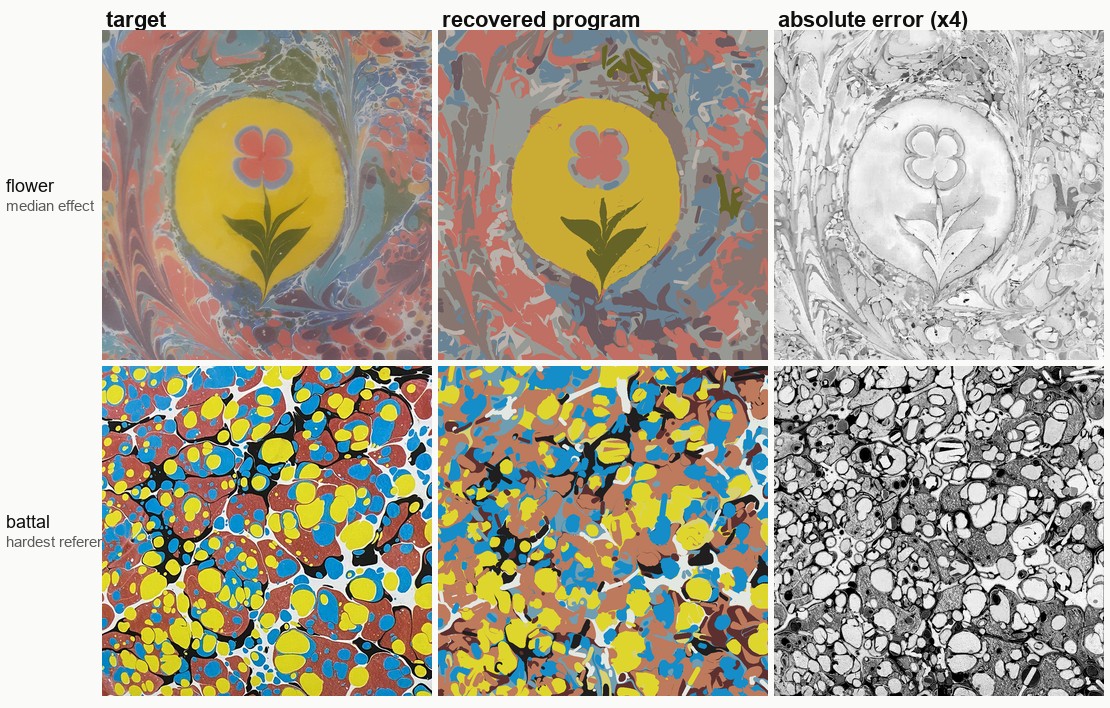}
\caption{\textbf{Where the residual is.} Target, replayed program, and
absolute error amplified $4\times$ (dark is larger error), at $1024^2$. Error
concentrates on structures near the stroke-width floor rather than being
distributed over the canvas.}\label{fig:errmaps}
\end{figure}

\paragraph{Containment.} Reparameterising each anchor by the current width and
offset can confine every footprint to a centred window covering $64\%$ of the
canvas. The constraint holds for every stroke in all five programs and changes
hard-render MSE by a median $+6.5\%$ (per reference $+3.7\%$ to $+15.0\%$).
That median is comparable to the across-seed spread of
Table~\ref{tab:noise}; flower has the largest increase, $+15.0\%$. It constrains where the tool
may act, not where transport later carries the pigment: the constrained programs still retain $78$--$79\%$ of their painted area outside the
canvas after transport, as the unconstrained ones do.

\paragraph{Area budget versus stroke count.}
Deposited area is a resource, and the area law makes it a controllable one:
the solver carries a penalty $\lambda_A\sum_i A_{\ell_i}(r_i)$ with the exact
gradient $\partial_r A_\ell=P_\ell$. Sweeping $\lambda_A$ from $0$ to $1$ on
all five references cuts deposited area by $83\%$, from $6.8$ to $1.2$ canvas
areas, at a $5\%$ reduction in the geometric-mean floor ratio, with
per-reference changes from $-13\%$ to $+5\%$. What the penalty does \emph{not}
do is remove a stroke: at $\lambda_A=1$ every one of the $2000$ still covers
more than a pixel and none sits at the radius floor. The loss is temporal
instead --- the fraction of strokes still visible rises from $6\%$ in the
first decile of the sequence to $99\%$ in the last, and $83\%$ of all
deposited paint never reaches the image. Most of it is not covered but
expelled: footprints are laid almost entirely on the canvas ($2\%$ outside at
deposition), and it is the later insertions that carry the early paint off it.

The budget that removes strokes is the count $N$, and
Table~\ref{tab:dropn} re-scores the archived count sweep to show what it
trades. As $N$ falls from $4000$ to $125$ the visible fraction rises
monotonically from $45\%$ to $100\%$, while the floor ratio falls from
$1.75\times$ to $0.91\times$: a $125$-stroke program is \emph{worse} than a
flat colour on these references. At $N=500$, $90\%$ of strokes are visible at $1.39\times$ the floor; at $N=2000$, $60\%$ at $1.61\times$. Fidelity and visibility trade off along this axis, and the operating
point of the paper sits on the fidelity side of it. Both budgets count \emph{deposited} area, what a physical bath consumes; delivered area on the canvas is a smaller quantity.

\begin{table}[t]\centering\footnotesize\setlength{\tabcolsep}{3.5pt}
\begin{tabular}{r rr rr rr}\toprule
& \multicolumn{2}{c}{deposited area} & \multicolumn{2}{c}{floor$/$MSE} & \multicolumn{2}{c}{visible strokes}\\
\cmidrule(lr){2-3}\cmidrule(lr){4-5}\cmidrule(lr){6-7}
$\lambda_A$ & GM (canvas) & vs.\ $0$ & GM & vs.\ $0$ & GM & vs.\ $0$\\
\midrule
0.00 & 6.76 & +0\% & 1.75$\times$ & +0\% & 1215 & +0\\
0.03 & 2.50 & -63\% & 1.74$\times$ & -0\% & 1217 & +2\\
0.30 & 1.51 & -78\% & 1.69$\times$ & -3\% & 1267 & +52\\
1.00 & 1.16 & -83\% & 1.66$\times$ & -5\% & 1353 & +138\\
\bottomrule
\end{tabular}
\caption{Paint budget as a penalty $\lambda_A\sum_i A_{\ell_i}(r_i)$ on
deposited area, geometric means over the five references. This series was
trained under the deep-feature objective with a free radius window and no
length penalty, so it is read within itself, like Sec.~\ref{sec:abl}; scores
are pixel MSE on the reference renderer. At $\lambda_A=1.0$ deposited area falls $83\%$ while the floor ratio changes $-5\%$ in geometric mean,
with nonuniform changes across references; per-reference changes
range from $-13\%$ (teardrops) to $+5\%$. The headline seed spread does
not establish significance for this separately trained series. \emph{Visible strokes} counts those owning at least
$50$ pixels of the final image: unchanged at $\lambda_A=0.03$ and rising on
every reference from $\lambda_A=0.3$, by $+83$ to $+184$ per program at
$\lambda_A=1.0$.}\label{tab:budget}
\end{table}

\begin{table}[t]\centering\footnotesize\setlength{\tabcolsep}{2.5pt}
\begin{tabular}{r r r r r r}\toprule
$N$ & area & floor$/$MSE & $\Delta$ & visible & vis.$/N$\\
\midrule
125 & 0.76 & 0.91$\times$ &  & 124 & 100\%\\
250 & 1.13 & 1.12$\times$ & +23\% & 242 & 97\%\\
500 & 2.12 & 1.39$\times$ & +24\% & 452 & 90\%\\
1000 & 3.49 & 1.51$\times$ & +9\% & 779 & 78\%\\
2000 & 5.77 & 1.61$\times$ & +7\% & 1218 & 61\%\\
4000 & 10.44 & 1.75$\times$ & +8\% & 1788 & 45\%\\
\bottomrule
\end{tabular}
\caption{The stroke count as a budget, from the archived count sweep, restricted to the two
references of that sweep still in the corpus (flower and battal) (deep-feature objective; a separate series, read
within itself), re-scored on the reference renderer at $1024^2$. Unlike the
area penalty of Table~\ref{tab:budget}, which shrinks every stroke and
removes none, lowering $N$ removes strokes outright. \emph{Visible} counts
strokes owning at least $50$ pixels of the final image and
\emph{visible$/N$} is the fraction of the program that can still be matched
to the image: it runs from
$100\%$ at $N=125$ to $45\%$ at $N=4000$, while the visible \emph{count}
runs from $124$ to $1788$. At $N=125$ the program is \emph{worse} than a flat
colour ($0.91\times$). Fewer strokes means a larger share of the program
survives to the image and a smaller absolute number does.}\label{tab:dropn}
\end{table}

\paragraph{Penalising non-visible pigment.} The
area penalty charges every stroke, but the paint that never reaches the image
is a specific subset, and the compositing identity $\sum_i v_i=\sum_x(1-T_0)$,
with $v_i=\sum_x T_i\alpha_i$ the area stroke $i$ owns in the image, shows that its complement, the visible area, can be computed with one additional render using unit colours. We
therefore added $\lambda_O\big(\sum_i A^\tau_i-\text{visible area}\big)$ and
re-fit all five references with a free radius window
(Table~\ref{tab:occ}). In this regime the term does not behave as intended: the canvas is painted almost entirely in every program, so shrinking a
stroke almost never exposes background, and the penalty acts approximately as
a total-area penalty with an offset. At $\lambda_O=1$ the deposited total
falls from $5.8$ to $1.3$ canvas areas and the unseen share from $83\%$ to
$24\%$, while the visible fraction of strokes rises from $59\%$ to $67\%$.
The free-radius control alone improves the geometric-mean floor ratio from
$3.67\times$ to $3.98\times$, so the headline radius window restricts
attainable fidelity; against that control the aggregate change stays within
$1\%$, with teardrops losing $12.9\%$ at $\lambda_O=1$. No setting removes a
stroke, and replacing the fixed count by an error target ($1.05\times$ the
headline error) with an adaptive multiplier on the data term holds the
requested quality ($3.70\times$) and also removes no strokes. Only the stroke count $N$ reduces the number of strokes.

\begin{table}[t]\centering\scriptsize\setlength{\tabcolsep}{2.5pt}
\begin{tabular}{l rr rr r}\toprule
arm & floor$/$MSE & vs.\ headline & unseen & visible & $N_{\mathrm{eff}}$\\
\midrule
control, $\lambda_O=0$ & 3.98$\times$ & +8.4\% & 83\% & 59\% & 2000\\
$\lambda_O=0.1$ & 4.01$\times$ & +9.3\% & 46\% & 60\% & 2000\\
$\lambda_O=0.3$ & 3.95$\times$ & +7.6\% & 34\% & 62\% & 2000\\
$\lambda_O=1$ & 3.95$\times$ & +7.6\% & 24\% & 67\% & 2000\\
error target, $\lambda_O=0.3$ & 3.70$\times$ & +0.8\% & 40\% & 62\% & 2000\\
\bottomrule
\end{tabular}
\caption{Penalising deposited area not visible in the final canvas. All five references, headline
objective and step budget, free radius window, one run per cell, scored on
the reference renderer at $1024^2$; geometric means and plain means over the five. \emph{Unseen} is the fraction of deposited paint that never reaches the
image, whether pushed off the canvas by later insertions or covered on it
(footprints are laid almost entirely on the canvas, $2\%$ outside at
deposition; the archived off-canvas measurement of Sec.~\ref{sec:process}
shows that most of the unseen paint is expelled rather than covered);
\emph{visible} the fraction of strokes owning at least $50$ pixels;
$N_{\mathrm{eff}}$ the smallest number of strokes exceeding one pixel of area
in any program. The free-radius control improves on the headline in these single-run
comparisons; the headline used the radius window of Table~\ref{tab:hyper}. With the canvas painted almost entirely (visible area
$0.97$--$0.99$ in every arm) the visible term varies little and the penalty acts approximately as a
total-area penalty. Against the control, aggregate image-error changes are
small, with a larger loss on teardrops at $\lambda_O=1$. No tested setting
removes a stroke; these are descriptive comparisons.}\label{tab:occ}
\end{table}

\section{Ablations}\label{sec:abl}
The stroke-count series uses two references and a fixed deep-feature data
term; the stroke-type comparison uses all five. Each comparison is internal
to a matched series, and all values are hard-rendered.

\paragraph{Stroke count.} Table~\ref{tab:count}: reconstruction improves
monotonically in $N$ on both references with overall diminishing returns, each doubling yielding improvements of $5.3$, $6.1$, $5.3$, $4.6$ and $3.4\%$. This series is
scored in its own deep-feature units, so the comparable noise figure is the
across-seed spread measured in those units, $0.5$--$0.6\%$, not the pixel-MSE
band of Table~\ref{tab:noise}; every doubling listed is outside it. Returns
diminish but remain resolvable at the last step, so $N=2000$ is an operating
point chosen for cost rather than a measured ceiling.

\paragraph{Continuous stroke type.} At $\ell_i=0$ the operator reduces exactly
to the classical drop map \cite{Lu2012}, so fixing $\ell_i=0$ for every stroke
is a fit under the drop dynamics alone. It increases loss
on all five references by a median $6.2\%$ on that series' own deep-feature
objective, against an across-seed spread of $0.5$--$0.6\%$ in the same units,
so both the sign and the magnitude are resolvable. Re-fitted under the
headline settings themselves (pixel-MSE objective, same seed and budget;
supplement), fixing $\ell_i=0$ raises pixel MSE on 4/5 references by a
geometric mean of $+7.8\%$ ($-1.7$ to $+15.0\%$, against a re-run band of
$2.7$--$6.7\%$) and LPIPS-Alex on 5/5 by $+9.5$ to $+25.0\%$; the elongated
deposits --- stems, leaves, the tails of the teardrops --- become chains of
drops (supplement), an expressiveness cost that the perceptual metric registers on
every reference and the pixel metric only partly. In
the recovered programs, $78\%$ of strokes have $\ell_i\ge0.004$ (four pixels
at $1024^2$) and $59\%$ have $\ell_i>r_i$, so the optimiser makes substantive
use of the continuous drop-to-line degree of freedom. The comparison is between stroke \emph{shapes} within our vocabulary; the operator set of \cite{Lu2012} also includes zero-injection tine and stylus maps, which lie outside it.

\begin{table}[t]\centering\footnotesize\setlength{\tabcolsep}{3.5pt}
\begin{tabular}{r rr r r}\toprule
$N$ & flower & battal & mean & per doubling\\
\midrule
125 & 5.1723 & 5.9314 & 5.552 & \\
250 & 4.9600 & 5.5528 & 5.256 & -5.3\%\\
500 & 4.7266 & 5.1401 & 4.933 & -6.1\%\\
1000 & 4.5795 & 4.7669 & 4.673 & -5.3\%\\
2000 & 4.3969 & 4.5163 & 4.457 & -4.6\%\\
4000 & 4.3113 & 4.2997 & 4.306 & -3.4\%\\
\bottomrule
\end{tabular}
\caption{Reconstruction versus stroke count under the deep-feature training
objective of that series (lower is better; the two development references in the corpus).
Monotone with diminishing returns; the last doubling buys $3.4\%$, which is
still outside the $0.5$--$0.6\%$ across-seed spread measured in these same
deep-feature units.}\label{tab:count}
\end{table}

\section{Generality of the construction}\label{sec:transfers}
The marbling-specific components of this construction can be separated from the generic ones.

\paragraph{Marbling-specific.} Only the operator itself: the choice to foliate
by distance to the segment, the Steiner area law \eqref{eq:steiner}, and the
symmetry argument that fixes the tangential coordinate. Another medium would
need its own $\psi$, and Sec.~\ref{sec:fidelity} is the template for quantifying the discrepancy between that $\psi$ and a
reference transport model.

\paragraph{Generic.} The distinction in Definition~\ref{def:aware} identifies
whether later strokes alter the coordinates at which earlier coverage is
evaluated. This makes image fitting transport-dependent, but does not require
inversion for differentiation: ordinary automatic differentiation and
checkpointing remain alternatives.

Replay applies when intermediate states can be reconstructed with computable
inverses and derivatives away from explicitly detectable exceptional sets.
The non-reconstructible state must be retained, and numerical conditioning may
require checkpoints. Closed-form inversion is useful but not necessary (supplement, S6); the memory benefit depends on the actual record
count and checkpoint schedule. Extending this approach to another wet medium
requires a suitable state representation and transport law. Diffusion,
mixing and pigment-dependent dynamics introduce additional state loss and
ill-conditioning, and establishing replay efficiency there is a separate
problem.

\section{Limitations}\label{sec:limits}
\paragraph{Scope of the evidence.} The evaluation contains five references
from one medium and one primary optimisation run per cell. Two of them were
used during development. The large paired effects reproduce on the three
held-out images and exceed measured numerical variation, but a larger,
multi-seed corpus is needed to estimate performance across marbling schools
and acquisition conditions. The comparison set is mostly internal: the
inert-mark arm, the one-pass correction and the second transport operator each
isolate one mechanism of the model; the one external system, a stroke-based fitter,
optimises inert marks and is scored only as a raster; a shared task on which both kinds
of system execute, and execution on a physical bath, are future work.

\paragraph{Model.} The operator is a geometric surrogate: exactly
area-preserving off the footprint, but with the foliation and tangential
transport of Sec.~\ref{sec:model} rather than those of a harmonic potential,
with a $27\%$ discrepancy from line-source potential flow at the median aspect ratio under \eqref{eq:D} ($8\%$ for the calibrated variant of Sec.~\ref{sec:calibrated}). It uses a free plane with no tray wall, passive pigment,
quasi-static sequential strokes, and no viscosity contrast, fingering,
interfacial tension, or diffusion. The recovered programs inject
$5.5$--$6.2$ canvas areas, much of which leaves the visible window; direct
physical execution therefore requires a bounded-domain model and a calibrated
paint budget. Zero-injection stylus drag is also outside the present stroke
family. The transfer experiment supports variation within the insertion class and
says nothing about arbitrary fluid dynamics.

\paragraph{Depth.} Measured step time grows approximately linearly with
program length (Sec.~\ref{sec:adjoint}), and replay memory includes
$P\lceil N/k\rceil$ checkpoints and $S$ exception records. We have not run
$N=10^4$. At fixed optimisation steps, increasing the number of parameters
can also change convergence. Table~\ref{tab:count} shows diminishing gains
through $N=4000$; it does not identify the limiting resource at greater depth.

\paragraph{Zero-injection transport.} Tine and stylus drag
\cite{Lu2012,Jaffer2018} is a different displacement class with no
deposition set. We report no calibrated comparison against it; the analytical brushes of \cite{James2026} are a published candidate evaluator for that class.

\paragraph{Identifiability.} Program parameters have ambiguities, including
palette-label permutations. Strokes below the ownership threshold can still
influence the final image through transport, so low visibility does not imply
absence of information about them. Our synthetic matching experiment assesses
one recovery procedure and positional statistic. We seek a program attaining low reconstruction loss under the model, not the historical process that made a sheet. Stroke count
and ordered slots are fixed; discrete program search and physical calibration
remain separate problems.

\section{Conclusion}
We presented a transport-coupled stroke design and the inverse-graphics
solution under it: the recovery of executable programs of such strokes from
images, for a digital marbling model. A capsule primitive joins drops to
elongated deposits, with exactly area-preserving exterior transport and a
closed-form exterior inverse. Its replay adjoint regenerates intermediate
states while retaining exceptional coordinates and position checkpoints.
At the reported operating point, the PyTorch implementation uses
$8.7\times$ less memory than the tested checkpointed-autograd baseline at
comparable step time; the fused implementation fits $2000$-stroke programs
in minutes. The programs replay at multiple resolutions and support edits
under the same dynamics. Transport-aware fitting outperforms the tested
transport-disabled and one-pass compensation controls when replayed under
the insertion model and, scored as a raster, is level with a published
stroke-based fitter at matched stroke count. Alternative-transport and synthetic experiments delimit
operator sensitivity and generating-stroke correspondence. The result is an
editable computational process for a medium in which later actions move
earlier marks.

\paragraph{Availability.} The recovered programs, the recorded scores from
which every table is generated, the generating scripts, the discrepancy
computation of Sec.~\ref{sec:fidelity} and the reference implementation of
the forward model are archived and available from the corresponding author
on request. The source package includes the supplementary PDF and the
table-generation script; regeneration of the measurements requires the
separate results archive. A public code-and-data release accompanies acceptance; during review the
archive is available to the reviewers through the editor.

% Bibliography.


\begin{thebibliography}{99}\small

% ---------- Stroke-based rendering: classical ----------

\bibitem{Hertzmann2003} A.~Hertzmann. A Survey of Stroke-Based Rendering.
  \emph{IEEE Computer Graphics and Applications}, 23(4):70--81, 2003.

% ---------- Stroke-based rendering: neural / differentiable ----------

\bibitem{Huang2019} Z.~Huang, W.~Heng, S.~Zhou. Learning to Paint with
  Model-based Deep Reinforcement Learning. In \emph{Proc. ICCV}, 2019.
\bibitem{Li2020} T.-M.~Li, M.~Luk\'a\v{c}, M.~Gharbi, J.~Ragan-Kelley.
  Differentiable Vector Graphics Rasterization for Editing and Learning.
  \emph{ACM Transactions on Graphics (Proc. SIGGRAPH Asia)}, 39(6), 2020.
\bibitem{Zou2021} Z.~Zou, T.~Shi, S.~Qiu, Y.~Yuan, Z.~Shi. Stylized Neural
  Painting. In \emph{Proc. CVPR}, 2021.
\bibitem{Liu2021} S.~Liu, T.~Lin, D.~He, F.~Li, R.~Deng, X.~Li, E.~Ding,
  H.~Wang. Paint Transformer: Feed Forward Neural Painting with Stroke
  Prediction. In \emph{Proc. ICCV}, 2021.

% ---------- Feedback-based medium-awareness (robot painting) ----------

\bibitem{Schaldenbrand2023} P.~Schaldenbrand, J.~McCann, J.~Oh. FRIDA: A
  Collaborative Robot Painter with a Differentiable, Real2Sim2Real Planning
  Environment. In \emph{Proc. ICRA}, 2023.

% ---------- Physically based paint and ink media (placed paint moves) ----------
\bibitem{Curtis1997} C.~J.~Curtis, S.~E.~Anderson, J.~E.~Seims,
  K.~W.~Fleischer, D.~H.~Salesin. Computer-Generated Watercolor. In
  \emph{Proc. SIGGRAPH}, 1997.

% ---------- Marbling: mathematical and physical ----------
\bibitem{Lu2012} S.~Lu, A.~Jaffer, X.~Jin, H.~Zhao, X.~Mao. Mathematical
  Marbling. \emph{IEEE Computer Graphics and Applications}, 32(6):26--35,
  2012.
\bibitem{Acar2006} R.~Acar, P.~Boulanger. Digital Marbling: A Multiscale
  Fluid Model. \emph{IEEE Transactions on Visualization and Computer
  Graphics}, 12(4):600--614, 2006.

\bibitem{Jaffer2018} A.~Jaffer. Oseen Flow in Paint Marbling.
  arXiv preprint 1702.02106, 2017.

% ---------- Inverse design through fluid dynamics ----------

% ---------- Differentiable rendering: relaxation and adjoints ----------
\bibitem{LiuSoftRas2019} S.~Liu, T.~Li, W.~Chen, H.~Li. Soft Rasterizer: A
  Differentiable Renderer for Image-based 3D Reasoning. In \emph{Proc. ICCV},
  2019.

\bibitem{Vicini2021} D.~Vicini, S.~Speierer, W.~Jakob. Path Replay
  Backpropagation: Differentiating Light Paths using Constant Memory and
  Linear Time. \emph{ACM Transactions on Graphics (Proc. SIGGRAPH)}, 40(4),
  2021.

% ---------- Memory-efficient differentiation ----------
\bibitem{Griewank2000} A.~Griewank, A.~Walther. Algorithm 799: Revolve: An
  Implementation of Checkpointing for the Reverse or Adjoint Mode of
  Computational Differentiation. \emph{ACM Transactions on Mathematical
  Software}, 26(1):19--45, 2000.

\bibitem{Gomez2017} A.~N.~Gomez, M.~Ren, R.~Urtasun, R.~B.~Grosse. The
  Reversible Residual Network: Backpropagation Without Storing Activations.
  In \emph{Proc. NeurIPS}, 2017.

% ---------- Recent stroke reconstruction (concurrent / post-2023) ----------
\bibitem{Jiang2025} Y.~Jiang, J.~Lu, Y.~Chen, Y.~He, K.~Wu, Y.~Yang,
  C.~Jiang. Birth of a Painting: Differentiable Brushstroke Reconstruction.
  arXiv:2511.13191, 2025.

\bibitem{LiAdjoint2025} Z.~Li, J.~He, B.~B\"orcs\"ok, T.~Zhang, D.~Chen,
  T.~Du, M.~C.~Lin, G.~Turk, B.~Zhu. An Adjoint Method for Differentiable
  Fluid Simulation on Flow Maps. In \emph{Proc. ACM SIGGRAPH Asia}, 2025.
\bibitem{Shu2026} L.~Shu, Y.~Jiang, K.~Wu, Y.~Yang, L.~Guibas, C.~Jiang.
  Differentiate the Solver, Not the Equation: Reverse-Sweep Adjoints for
  Block Implicit Simulation. arXiv:2608.08559, 2026.

% ---------- Metric ----------
\bibitem{Zhang2018} R.~Zhang, P.~Isola, A.~A.~Efros, E.~Shechtman, O.~Wang.
  The Unreasonable Effectiveness of Deep Features as a Perceptual Metric. In
  \emph{Proc. CVPR}, 2018.

% Compositing and differentiable rasterisation provenance for Eqs. (3) and (9)
\bibitem{Porter1984} T.~Porter, T.~Duff. Compositing Digital Images. In
  \emph{Proc. SIGGRAPH}, 1984.

% Reverse-time state reconstruction and its failure for non-reversible dynamics

\bibitem{Maclaurin2015} D.~Maclaurin, D.~Duvenaud, R.~P.~Adams.
  Gradient-based Hyperparameter Optimization through Reversible Learning. In
  \emph{Proc. ICML}, 2015.
% Adjoints through simulation chains in graphics
\bibitem{McNamara2004} A.~McNamara, A.~Treuille, Z.~Popovi\'c, J.~Stam. Fluid
  Control Using the Adjoint Method. \emph{ACM Transactions on Graphics (Proc.
  SIGGRAPH)}, 23(3), 2004.
\bibitem{Hu2020} Y.~Hu, L.~Anderson, T.-M.~Li, Q.~Sun, N.~Carr,
  J.~Ragan-Kelley, F.~Durand. DiffTaichi: Differentiable Programming for
  Physical Simulation. In \emph{Proc. ICLR}, 2020.

% ---------- Analytical fluid brushes ----------
\bibitem{James2026} D.~L.~James, E.~James. Mixwell: Sharp 2D Fluid Brushes
  for Progressive Physics-Based Mixing. \emph{ACM Transactions on Graphics
  (Proc. SIGGRAPH)}, 45(4), 2026.

% ---------- Process recovery from a finished artwork ----------
\bibitem{Fu2011} H.~Fu, S.~Zhou, L.~Liu, N.~J.~Mitra. Animated Construction
  of Line Drawings. \emph{ACM Transactions on Graphics (Proc. SIGGRAPH
  Asia)}, 30(6):133, 2011.
\bibitem{Tan2015} J.~Tan, M.~Dvoro\v{z}\v{n}\'ak, D.~S\'ykora, Y.~Gingold.
  Decomposing Time-Lapse Paintings into Layers. \emph{ACM Transactions on
  Graphics (Proc. SIGGRAPH)}, 34(4):61, 2015.

\bibitem{Ma2022} X.~Ma, Y.~Zhou, X.~Xu, B.~Sun, V.~Filev, N.~Orlov, Y.~Fu,
  H.~Shi. Towards Layer-wise Image Vectorization. In \emph{Proc. CVPR}, 2022.

% ---------- Geometry of parallel sets ----------
\bibitem{Federer1959} H.~Federer. Curvature Measures. \emph{Transactions of
  the American Mathematical Society}, 93(3):418--491, 1959.

% ---------- Image-guided mark placement: origins ----------
\bibitem{Haeberli1990} P.~Haeberli. Paint by Numbers: Abstract Image
  Representations. In \emph{Proc. SIGGRAPH}, 1990.
\bibitem{Hertzmann1998} A.~Hertzmann. Painterly Rendering with Curved Brush
  Strokes of Multiple Sizes. In \emph{Proc. SIGGRAPH}, 1998.

% ---------- Gradients across visibility discontinuities ----------
\bibitem{LiEdge2018} T.-M.~Li, M.~Aittala, F.~Durand, J.~Lehtinen.
  Differentiable Monte Carlo Ray Tracing through Edge Sampling. \emph{ACM
  Transactions on Graphics (Proc. SIGGRAPH Asia)}, 37(6), 2018.
\bibitem{Loubet2019} G.~Loubet, N.~Holzschuch, W.~Jakob. Reparameterizing
  Discontinuous Integrands for Differentiable Rendering. \emph{ACM Transactions
  on Graphics (Proc. SIGGRAPH Asia)}, 38(6), 2019.

% ---------- Programs recovered from images ----------
\bibitem{Ellis2018} K.~Ellis, D.~Ritchie, A.~Solar-Lezama, J.~B.~Tenenbaum.
  Learning to Infer Graphics Programs from Hand-Drawn Images. In \emph{Proc.
  NeurIPS}, 2018.
\bibitem{Talton2011} J.~O.~Talton, Y.~Lou, S.~Lesser, J.~Duke, R.~M\v{e}ch,
  V.~Koltun. Metropolis Procedural Modeling. \emph{ACM Transactions on
  Graphics}, 30(2), 2011.

% ---------- The medium ----------
\bibitem{Wolfe1990} R.~J.~Wolfe. \emph{Marbled Paper: Its History, Techniques,
  and Patterns}. University of Pennsylvania Press, 1990.

\bibitem{Dinh2015} L.~Dinh, D.~Krueger, Y.~Bengio. NICE: Non-linear Independent
  Components Estimation. In \emph{ICLR Workshop Track}, 2015. arXiv:1410.8516.
\end{thebibliography}
\end{document}